\documentclass[conference]{IEEEtran}
\usepackage[
    type={CC},
    modifier={by-sa},
    version={4.0},
]{doclicense}
\IEEEoverridecommandlockouts
\ifCLASSINFOpdf
\else
\fi
\usepackage{amsmath,amsthm,amssymb}
\makeatletter
\RequirePackage[bookmarks,unicode,colorlinks=true]{hyperref}%
   \def\@citecolor{blue}%
   \def\@urlcolor{blue}%
   \def\@linkcolor{blue}%

\def\orcidID#1{\smash{\href{http://orcid.org/#1}{\protect\raisebox{-1.25pt}{\protect\includegraphics{orcid_color.eps}}}}}
\makeatother

\usepackage[capitalise]{cleveref}

\theoremstyle{plain}
\newtheorem{theorem}{Theorem}
\newtheorem{lemma}[theorem]{Lemma}
\newtheorem{proposition}[theorem]{Proposition}
\newtheorem{corollary}[theorem]{Corollary}

\makeatletter  
\def\@endtheorem{\qed\endtrivlist\@endpefalse } 
\makeatother
\theoremstyle{definition}
\newtheorem{remark}[theorem]{Remark}
\newtheorem{example}[theorem]{Example}
\newtheorem{definition}[theorem]{Definition}

\usepackage{bbding}
\usepackage{graphicx}

\usepackage{graphicx,multirow,tabularx,makecell}
\usepackage{microtype}
\usepackage[compact]{titlesec}
\usepackage{paralist}
\usepackage[all]{xy}
\usepackage{stmaryrd}
\usepackage{ebproof,proof} 

\RequirePackage{symbolic-bisim}
\RequirePackage{more-macros}

\RequirePackage{mathpartir}

\usepackage{versions}
\excludeversion{exclproof}
\excludeversion{technical}
\newcommand\mytechnical[1]{}

\usepackage[misc]{ifsym}

\newcommand\cutout[1]{}

\begin{document}
%
\title{Fully Abstract Normal Form Bisimulation for Call-by-Value PCF\thanks{This publication has emanated from research supported in part by a grant from Science Foundation Ireland under Grant number 13/RC/2094\_2.%
\\[1ex] This version of the paper has undergone minor corrections since its original publication in LICS 2023.} %
}

\author{\IEEEauthorblockN{Vasileios Koutavas}
 \IEEEauthorblockA{Trinity College Dublin\\
 Dublin, Ireland\\
 Email: Vasileios.Koutavas@tcd.ie}
 \and
 \IEEEauthorblockN{Yu-Yang Lin}
 \IEEEauthorblockA{Trinity College Dublin\\
 Dublin, Ireland\\
 Email: linhouy@tcd.ie}
 \and
 \IEEEauthorblockN{Nikos Tzevelekos}
 \IEEEauthorblockA{Queen Mary University of London\\
 London, UK\\
 Email: Nikos.Tzevelekos@qmul.ac.uk}}

\maketitle              

\begin{abstract}
  We present the first fully abstract normal form bisimulation for call-by-value PCF (PCF$_{\textsf{v}}$). 
  Our model is based on a labelled transition system (LTS) that combines elements from applicative bisimulation, environmental bisimulation and game semantics. In order to obtain completeness while avoiding the use of semantic quotiening, the LTS constructs 
  traces corresponding to interactions with possible functional contexts.  
The model gives rise to a sound and complete 
  technique for checking of PCF$_{\textsf{v}}$ program equivalence,
  which we implement in a bounded bisimulation checking tool. We test our tool on known equivalences from the literature and new examples.
\end{abstract}

\section{Introduction}
\label{sec:intro}
The full abstraction problem for PCF, i.e.\ constructing a denotational model that captures contextual equivalence in the paradigmatic functional language PCF, was put forward by Plotkin in the mid 1970's~\cite{plotkin77}.
The first fully abstract denotational models for PCF were presented in the early 1990's and gave rise to the theory of \emph{game semantics}~\cite{AJM,HO,Nickau}, while fully abstract models for its call-by-value variant were given in~\cite{PCFv,PCFvb}.
Fully abstract operational models of PCF have been given in terms of \emph{applicative bisimulations} \cite{applicative90,GordonApplicative95,Howe96} and \emph{logical relations} \cite{OHearn95}, and for other pure languages in terms of \emph{environmental bisimulations} \cite{SumiiPierce07b,environmental} and \emph{logical relations} \cite{Pitts98LR,step-indexed}.
On the other hand, Loader demonstrated that contextual equivalence for finitary PCF is undecidable~\cite{Loader01}.

A limitation of the game semantics models for PCF is their intentional nature.
While the denotations of inequivalent program terms are always distinct, there are equivalent terms whose denotations are also distinct and become equivalent only after a semantic quotiening operation.
Quotiening requires universal quantification over tests, which amounts to quantification over all (innocent) contexts.
This hinders the use of game models for pure functional languages to prove equivalence of terms, as any reasoning technique needs to involve all contexts of term denotations in the semantic model (i.e.\ all possible \emph{Opponent strategies}).
In a more recent work, Churchill et al.~\cite{ChurchillEtal2020} were able to give a direct characterisation of program equivalence in terms of so-called \emph{sets of O-views}, built out of term denotations. 
The latter work is to our knowledge the only direct (i.e.\ quotient-free) semantic characterisation of PCF contextual equivalence, though it is arguably more of theoretical value and does not readily yield a proof method.

Operational models also involve quantification over all identical (applicative bisimulation) or related (logical relations, environmental bisimulations) closed arguments of type $A$, when describing the equivalence class of type $A \to B$.
Although successful proof techniques of equivalence have been developed based on these models, universal quantification over opponent-generated terms must be handled in proofs with rather manual inductive or coinductive arguments.

Normal-Form (NF) bisimulation, also known as open bisimulation, was originally defined for characterising L\'evy-Longo tree equivalence for the lazy lambda calculus \cite{Sangiorgi:lazylambda} and adapted to languages with call-by-name \cite{LassenHNF}, call-by-value \cite{LassenENFB}, nondeterminism \cite{LassenNondet}, aspects \cite{JagadeesanPitcherRiely:aspects}, recursive types \cite{LassenLevy:nfbistypes}, polymorphism \cite{LassenLevy:nfbisimpoly}, control with state \cite{StorvingLassenPOPL07}, state-only \cite{BiernackiLP19}, and control-only \cite{BiernackiLengletNFB2}. It has also been used to create equivalence verification techniques for a lambda calculus with state \cite{hobbit}.

The main advantage of NF bisimulation is that it does away with quantification over opponent-generated terms, replacing them with fresh open names.
This has also been shown~\cite{LassenLevy:nfbistypes,LassenLevy:nfbisimpoly,hobbit} to relate to operational game semantics models where opponent-generated terms are also represented by names~\cite{Laird07,TzeGhica,Jaber15}.
The main disadvantage of NF bisimulation is that\,---\,with the notable exception of languages with control and state \cite{StorvingLassenPOPL07}, and state-only \cite{BiernackiLP19,hobbit}\,---\,it is too discriminating thus failing to be fully abstract with respect to contextual equivalence. This is particularly true for pure languages such as PCF, and its call-by-value variant \lang which is the target of this paper.

However, the discriminating power of NF bisimulation depends on the labelled transition system (LTS) upon which it is defined. 
Existing work define NF bisimulation over LTSs that treat call and return moves between term and context in a fairly standard way: these are immediately observable by the bisimulation as they appear in transition annotations, and context moves correspond to imperative, not purely functional, contexts.
As we show in \cref{sec:examples}, this is overly discriminating for a language such as \lang.
Moreover, existing NF bisimulation techniques, either do not make extensive use of the the context's knowledge in the LTS configurations (e.g.\ \cite{LassenLevy:nfbistypes}), or consider an ever-increasing context knowledge (e.g.\ \cite{BiernackiLP19}) which is only fully abstract for imperative contexts.

In this paper we present the first fully abstract NF bisimulation for \lang, defined in \cref{sec:lang}. To achieve this we develop in \cref{sec:lts} a novel Labelled Transition System (LTS) which:
\begin{itemize}
  \item is based on an operational presentation of game models (cf.~\cite{Laird07}) and uses Proponent and Opponent configurations (and \emph{call/return moves}) for evaluation steps that depend on the modelled term and its environment, respectively;
  \item uses an explicit stack principle to guarantee well-bracketing
    and  stipulates that the \emph{opponent view} of the LTS trace be restricted to moves related to the last open proponent call (cf.\ \emph{well-bracketing} and \emph{visibility}~\cite{HO});
  \item restricts opponent moves so that they
    correspond to those of a pure functional context, by explicitly keeping track of previous opponent responses to proponent moves and their corresponding (opponent) view (cf.\ \emph{innocence}~\cite{HO});
  \item postpones observations of proponent call/return moves until computations are complete to avoid unnecessary distinctions between equivalent terms.
\end{itemize}
We then define a notion of NF bisimulation over this LTS which combines standard move/label synchronisation with coherence of corresponding opponent behaviours.
We show that the latter is fully abstract with respect to contextual equivalence (\cref{sec:fa}).
Due to its operational nature and the absence of quantification over opponent-generated terms, the model lends itself to a bounded model-checking technique for equivalence which we implement in a prototype tool (\cref{sec:tool}). We conclude in \cref{sec:concl}.


\section{Motivating Examples}
\label{sec:examples}
We start with a simple example of equivalence that showcases unobservable behaviour differences in \lang.
\begin{example}\label{ex:pure}
Consider the following equivalent terms of type $(\Unit \to \Int) \to (\Unit \to \Int) \to \Int$.

\medskip
\noindent
\begin{tabular}{@{}l@{\;}l@{}}
$M_1 =$
&\begin{minipage}[t]{0.87\columnwidth}\vspace{-1.3em} \begin{lstlisting}[boxpos=t]
fun f$\;$->$\;$fun g$\;$->$\;$if f () == g () then
$\;$$\;$$\;$$\;$                if f () == g () then 0 else 1
$\;$$\;$$\;$$\;$                else 2
\end{lstlisting}\end{minipage}
\\
$N_1 =$ 
&\begin{minipage}[t]{0.87\columnwidth}\vspace{-1.3em} \begin{lstlisting}[boxpos=t]
fun f$\;$->$\;$fun g$\;$->$\;$if g () == f () then 0 else 2
\end{lstlisting}\end{minipage}
\end{tabular}

\medskip
These two terms are contextually equivalent because the context 
cannot observe whether \lstinline{f} and \lstinline{g} have been called more than once with the same argument.
Two calls of a pure and deterministic function with the same argument both diverge or return the same value. 
Moreover, the context cannot observe the order of calls to the context-generated functions \lstinline{f} and \lstinline{g}.
As we will see in \cref{sec:lts}, our LTS restricts the behaviour of context-generated functions such as \lstinline{f} and \lstinline{g} so that they behave in a pure deterministic manner, and does not make distinctions based on their call order.
\end{example}

We now discuss key \emph{observable} behaviour differences in \lang through the lens of bisimulation theories.
As explained in \cite{fromapplicative}, the main feature in environmental bisimulation definitions \cite{SumiiPierce07a,SumiiPierce07b,KoutavasW06,KW06esop,environmental} is \emph{knowledge accumulation:} environmental bisimulation collects the term-generated functions in an environment representing the knowledge of the context. This knowledge is used in the following bisimulation tests to distinguish terms:
 \begin{enumerate}
   \item\label{item:1} to call a function from the environment multiple times in a row with the same argument;
   \item\label{item:2} to call a function from the environment multiple times in a row with different arguments;
   \item\label{item:3} to call environment functions after other environment functions have returned; and 
   \item\label{item:4} to use environment functions in the construction of context-generated functions.
\end{enumerate}
The above is easily understood to be necessary in stateful languages and was shown to be needed in pure languages with existential types \cite{fromapplicative}. However, as applicative bisimulation has shown, it is unnecessary to accumulate the context's knowledge in order to create a theory of \lang: applicative bisimulation interrogates related functions in isolation from other knowledge by simply applying them to identical arguments.

As discussed in the first example of this section, purity and determinism indeed make (\ref{item:1}) unnecessary in \lang. 
However, (\ref{item:2}--\ref{item:4}) are \emph{necessary} tests that a normal form bisimulation theory for \lang must perform.
This is because a normal form bisimulation definition must prescribe the necessary interaction between terms and context \emph{under any evaluation context} and not just at top-level computations.
Applicative bisimulation on the other hand is only defined in terms of top-level function applications, where the context's knowledge is limited. Universal quantification over the code of context-generated function arguments implicitly encodes all the interactions that related terms may have with these arguments.
We showcase the need for (\ref{item:2}--\ref{item:4}) in the following three example inequivalences. 

\begin{example}\label{ex:u's}
  Consider the inequivalent terms $M_2$, $N_2$ of type\\
 $(((\Bool \to \Bool) * (\Bool \to \Bool)) \to \Bool) \to \Bool \to \Bool$.

\medskip
\noindent
\begin{tabular}{@{}l@{\;}l@{}}
$M_2 =$
&\begin{minipage}[t]{0.826\columnwidth}\vspace{-1.3em} \begin{lstlisting}[boxpos=t]
fun f -> fun b ->
  let rec X d = f (X, fun _ ->  d)
  in X b
\end{lstlisting}\end{minipage}
\\
$N_2 =$ 
&\begin{minipage}[t]{0.826\columnwidth}\vspace{-1.3em} \begin{lstlisting}[boxpos=t]
fun f -> fun b ->
  f ((fun _ -> _bot_), fun _ ->  b)
\end{lstlisting}\end{minipage}
\end{tabular}

\medskip
\noindent
  Here \lstinline{_bot_} is a diverging term and \lstinline{_} represents an unused variable;
\lstinline{X} has type $\Bool \to \Bool$.

Term $M_2$ will receive a function \lstinline{f} and a boolean \lstinline{b}.
It will then create a recursive term which calls \lstinline{f} with a pair containing two $\Bool\to\Bool$ functions.
If \lstinline{f} calls \lstinline{X}, the first function in the pair, with a boolean \lstinline{d}, computation will recur;
if it calls the second function, it will receive the argument of the latest call to \lstinline{X}.
On the other hand, $N_2$ calls \lstinline{f} with a pair of functions where the first one diverges upon call, and the second one returns \lstinline{b}, provided at the beginning of the interaction.

These terms can be distinguished by the following context:

\medskip
\noindent
\begin{tabular}{@{}l@{\;}l@{}}
\begin{minipage}[t]{0.99\columnwidth}\vspace{-1.3em} \begin{lstlisting}[boxpos=t]
let f = fun (X, fd) -> if fd false then X false
                       else true
in [] f true
\end{lstlisting}\end{minipage}
\end{tabular}

\medskip
\noindent
  This context creates a function \lstinline{f} that receives two functions \lstinline{X} and \lstinline{fd}, and conditionally calls \lstinline{X} with \lstinline{false}, if the call to \lstinline{fd} returns \lstinline{true}.
When placed in the hole \lstinline{[]} of this context, $M_2$ will receive \lstinline{f} and value \lstinline{true}.
Recursive function \lstinline{X} will thus be first called with \lstinline{true}, in the last line of $M_2$, and then again with \lstinline{false} by \lstinline{f}, causing the termination of the computation.
On the other hand, with $N_2$ in the hole, the context will diverge.

This is effectively the only simple context that can distinguish $M_2$ and $N_2$, and thus a NF bisimulation theory of equivalence for \lang must accumulate \lstinline{X} in the opponent's knowledge at inner interaction levels to allow calling \lstinline{X} after \lstinline{fd} has returned. This shows the need for allowing (\ref{item:3}) in a NF bisimulation.
Indeed, if we omit this from the technique we develop in the following sections, $M_2$ and $N_2$ would be deemed equivalent.
\end{example}

The following variation of the above example shows that the context may need to call the same function twice, with different arguments, to make observations.

\begin{example}\label{ex:v's}
  Consider the inequivalent terms $M_3$, $N_3$ of type\\
 $((\Bool \to (\Bool \to \Bool)) \to \Bool) \to \Bool \to \Bool$.

\medskip
\noindent
\begin{tabular}{@{}l@{\;}l@{}}
$M_3 =$
&\begin{minipage}[t]{0.87\columnwidth}\vspace{-1.3em} \begin{lstlisting}[boxpos=t]
fun f$\;$->$\;$fun b$\;$->
  let rec X d = f (fun e$\;$->$\;$if e then X 
                            else (fun _$\;$->$\;$d))
  in X b
\end{lstlisting}\end{minipage}
\\
$N_3 =$ 
&\begin{minipage}[t]{0.87\columnwidth}\vspace{-1.3em} \begin{lstlisting}[boxpos=t]
fun f$\;$->$\;$ fun b$\;$->
  f (fun e$\;$->$\;$if e then (fun d$\;$->$\;$_bot_) 
              else (fun _$\;$->$\;$b))
\end{lstlisting}\end{minipage}
\end{tabular}
%
%
%
%
%

\medskip
\noindent
where \lstinline{X} has type $ \Bool \to \Bool$.
The distinguishing context is

\medskip
\noindent
\begin{tabular}{@{}l@{\;}l@{}}
\begin{minipage}[t]{0.99\columnwidth}\vspace{-1.3em} \begin{lstlisting}[boxpos=t]
let f = fun fXd$\;$-> 
          let X = fXd true in
          let fd = fXd false in 
          if fd false then X false else true
in [] f true
\end{lstlisting}\end{minipage}
\end{tabular}

\medskip
\noindent
  Here the interaction between the terms and the context are as in the previous example, with the difference that the context must apply \lstinline{fXd} to \lstinline{true} and then \lstinline{false} to receive the two functions \lstinline{X} and \lstinline{fd}. The context terminates with $M_3$ but diverges with $N_3$ in its hole.

This is effectively the only simple context that can distinguish $M_3$ and $N_3$, and thus a NF bisimulation theory of equivalence for \lang that describes all the term-context interactions must accumulate \lstinline{fXd} in the context's knowledge in order to apply it twice in a row. This showcases the need for allowing (\ref{item:2}) in a NF bisimulation.
\end{example}

Our final example shows that functions from the context's knowledge must be used within a context-generated function in order to distinguish two terms.

\begin{example}\label{ex:fix-curried-ineq}
  Consider the inequivalent terms $M_4$, $N_4$ of type\\
 $T=((\Int \to \Int) \to \Int \to \Int) \to \Int \to \Int$.

\medskip
\noindent
\begin{tabular}{@{}l@{\;}l@{}}
$M_4 =$
&\begin{minipage}[t]{0.826\columnwidth}\vspace{-1.3em} \begin{lstlisting}[boxpos=t]
let rec X count = fun f -> fun i ->
  f  (fun j -> if (count > 0)
               then X (count-1) f j
               else _bot_)  i
in X $k$
\end{lstlisting}\end{minipage}
\\
$N_4 =$ 
&\begin{minipage}[t]{0.826\columnwidth}\vspace{-1.3em} \begin{lstlisting}[boxpos=t]
fun f -> fun i -> let rec Y j = f Y j
                  in Y i
\end{lstlisting}\end{minipage}
\end{tabular}
%
%
%
%
%

\medskip
\noindent
  where \lstinline{X} and \lstinline{Y} have type $\Int\to T$ and $ \Int \to \Int$, respectively.

This is a family of examples in which the distinguishing interaction increases with $k$; $N_4$ enables \lstinline{f} to call itself an arbitrary number of times,
whereas $M_4$ enables up to $k$ recursive calls of \lstinline{f} before it diverges.
The distinguishing context below attempts to perform $k+1$ recursive calls and then to return $0$:

\medskip
\noindent
\begin{tabular}{@{}l@{\;}l@{}}
\begin{minipage}[t]{0.99\columnwidth}\vspace{-1.3em} \begin{lstlisting}[boxpos=t]
[] (fun Z -> fun i -> if i > 0 then Z (i-1) else 0)
   $(k+1)$
\end{lstlisting}\end{minipage}
\end{tabular}

\medskip
\noindent
This context diverges with $M_4$ but converges with $N_4$ in its hole.
To achieve this, the context uses the function received as argument \lstinline{Z} inside the context-generated function \lstinline{fun i -> if i > 0 then Z (i-1) else 0} which is given back to the term.
  As this is effectively the only context that can distinguish $M_4$ and $N_4$, we need to allow our NF bisimulation for \lang to construct (symbolic) function values that can internally refer to functions in the context's knowledge at the time of construction; showing the need for allowing (\ref{item:4}) in a NF bisimulation.
If we omit this from our technique, it would deem $M_3$ and $N_3$ equivalent.
\end{example}


  \section{Language and Semantics}
  \label{sec:lang}
  \begin{figure*}[t]
  \[\begin{array}{r@{}r@{\;\,}c@{\;\,}l}
    {\Typ:} & T           & \mis & \Bool \mor \Int \mor \Unit \mor T \arrow T \mor T_1 * \ldots * T_n \\
      {\Exp: } &\;\; e,M,N           & \mis & v \mor  x\mor(\vec e)\mor \arithop{\vec e} \mor \app e e \mor \cond{e}{e}{e} 
 \mor \elet{(\vec x)}{e}{e} \\                                             
    {\Val:} & u,v       & \mis & c \mor\lam[f_T] x e   \mor (\vec v)                                                                          \\ 
    {\EC:} & E        & \mis & \hole_T \morcondensed (\vec v,E,\vec e) \morcondensed \arithop{\vec v,E,\vec e} \morcondensed \app E e \morcondensed \app v E \morcondensed \cond E e e \morcondensed \elet{(\vec x)}{E}{e} \\
    {\Cxt:} & D           & \mis & \hole_{i,T} \mor e \mor (\vec D) \mor \arithop{\vec D} \mor \app D D \mor \cond D D D \mor \lam[f_T] x D
                   \mor \elet{(\vec x)}{D}{D}       
  \end{array}\]
  \[\begin{array}{r@{\;\;}l@{\;\;}lr@{\;\;}l@{\;\;}ll}
      \app{(\lam[f] x e)}  v & \redbase & e\sub{x}{v}\sub{f}{\lam[f] x e}  &                                                                                 \arithop{\vec c}    & \redbase & w                      & \text{if } \mathop{op}^{\textsf{arith}}(\vec c) = w \\
    \elet{(\vec x)}{(\vec v)}{e} & \redbase & e\sub{\vec x}{\vec v}  &
    \cond{c}{e_1}{e_2}  & \redbase & e_i                    & \text{if } (c,i) \in \{(\true,1), (\false,2)\} \\ &&&
    E\hole[e]           & \red     & E\hole[e']            & \text{if } e \redbase e'
  \end{array}\]
  \hrule
  \caption{Syntax and reductions of \lang.
  Variables $x,y,z$,\,etc.\ sourced from countably infinite set \Var. $c$ ranges over constants $(),\true,\false$ and $\mathsf{n}$ (for any $n\in\mathbb{Z}$).}\label{fig:lang}
\end{figure*}

We work with the language \lang, a simply-typed call-by-value lambda calculus with boolean and integer operations~\cite{PCFv}.
The syntax and reduction semantics are shown in \cref{fig:lang}.
Expressions (\Exp) include the standard lambda expressions with recursive functions ($\lam x e$), together with  standard base type constants ($c$) and operations ($\arithop{\vec e}$), as well as conditionals and tuple-deconstructing let expressions ($\elet{(\vec x)}{e}{e}$).
We use standard macros, for example $\bot_T\defeq\lam[f_{\Unit\to T}]{x}{fx}$ and $\lambda x_T.e\defeq\lam[f_{T\to T'}]{x}{e}$ (with $f$ fresh for $e$). 

The language \lang is simply-typed with typing judgements of the form $\typing{\Delta}{e}{T}$, where $\Delta$ is a type environment (omitted when empty) and $T$ a value type (\Typ).
The rules of the typing system are standard and omitted here~\cite{PCFv}.
Values consist of boolean, integer, and unit constants, functions and arbitrary length tuples of values.

The reduction semantics is by small-step transitions between closed expressions, $e \red e'$,
defined using single-hole evaluation contexts ($\EC$) over a base relation $\redbase$.
Holes $\hole_T$ are annotated with the type $T$ of closed values they accept, which we omit when possible to lighten notation.
Beta substitution of $x$ with $v$ in $e$ is written as $e\sub{x}{v}$.
We write $e\trm$ to denote $e \red^* v$ for some $v$.
We write $\vec \chi$ to mean a finite sequence of syntax objects $\chi_1,\dots$, and assume standard syntactic sugar from the lambda calculus.
In our examples we assume an ML-like syntax and implementation of the type system, which is also the concrete syntax of our prototype tool (\cref{sec:tool}). 

Contexts $D$ contain multiple, non-uniquely indexed holes $\hole_{i,T}$, where $T$ is the type of value that can replace the hole (and each index $i$ can have one related type).
A context is called \emph{canonical} if its holes are indexed $1,\dots,n$, for some $n$.
Given a canonical context $D$ and a sequence of typed expressions $\Sigma\vdash\vec e:\vec T$,
notation $D\hole[\vec e]$ denotes the context $D$ with each hole $\hole_{i,T_i}$ replaced with $e_i$.
We omit hole types where possible and indices when all holes in $D$ are annotated with the same $i$.
Standard contextual equivalence~\cite{Morris68} follows.

\begin{definition}[Contextual Equivalence]\label{def:cxt-equiv}
Expressions $\vdash e_1:T$ and $\vdash e_2:T$ are \emph{contextually equivalent}, written as $e_1 \cxteq e_2 : T$, when for all contexts $D$ such that $\vdash D\hole[e_1]: \Unit$ and $\vdash D\hole[e_2]: \Unit$ we have
$
    D\hole[e_1]\trm ~\text{iff}~
    D\hole[e_2]\trm
$.
\end{definition}

Due to the language being purely functional, we can refine the contexts needed for contextual equivalence to \emph{applicative} ones.

\begin{definition}
  Applicative contexts are given by the syntax:
  \begin{align*}
    E_a &\mis \ \hole_T \mor  \app E_a v \mor \cond{E_a=c}{{()}}{\bot_\Unit}\mor \pi_i(E_a)
  \end{align*}
where $\pi_i(\chi)$ returns the $i$-th component of tuple $\chi$.
\end{definition}

Using the fact that applicative bisimulation is fully abstract~\cite{fromapplicative,Howe96}, we can show the following.

\begin{proposition}[Applicative contexts suffice]\label{prop:CL}
$e_1 \cxteq e_2:T$ iff for all applicative contexts $E_a$ such that $\vdash E_a\hole[e_1]_T,E_a\hole[e_2]_T: \Unit$ we have
$
    E_a\hole[e_1]\trm ~\text{iff}~
    E_a\hole[e_2]\trm
$.
\end{proposition}


  \section{LTS with Symbolic Higher-Order Transitions}
  \label{sec:lts}
   We now define a Labelled Transition System (LTS) which allows us to probe higher-order values with possible symbolic arguments.
 The LTS follows the operational game semantics approach, with several adjustments:
\begin{itemize}
\item the basis of the LTS is the operational game model of~\cite{Laird07};
\item the Opponent behaviours are constrained to \emph{innocent} ones (cf.~\cite{PCFv}) by use of an \emph{opponent memory} component $M$;

\item the denotation of an expression is not just the transitions that the LTS produces for this expression but, instead, these transitions together with the corresponding opponent memory at top-level configurations.
\end{itemize}
Thus, the LTS comprises of \emph{Proponent} and \emph{Opponent} configurations with corresponding transitions, modelling the computations triggered by an expression and its context respectively.
Opponent is construed as the syntactic context, which provides values for the functions that are open in the expression. Open functions are modelled with (opponent-generated) \emph{abstract names}, which are accommodated by extending the syntax and typing rules with abstract function names $\alpha$:
  $$
  \textsc{\Val: } \quad u,v,w  \, \mis \, c \mor \lam[f_T] x e\mor (\vec v) \mor \alpha_T^i
  $$
  Abstract function names $\alpha_{T}^i$ are annotated with the type $T$ of function they represent, and with an index $i\geq0$ that is used for bookkeeping; these are omitted where not important.
  $\an{\chi}$ is the set of abstract names in $\chi$.

The definition of our LTS (\cref{fig:lts}) is explained below.

\paragraph*{\flushleft Moves:}
Our LTS uses \emph{moves}:
$$
\eta \mis \lpropapp{\alpha_T}{D} \mor \lpropret{D} \mor \lopapp{i}{v} \mor \lopret{v}
$$
with contexts $D$ and values $v$ built from the following restricted grammars:
\begin{align*}
  D_\bullet &\mis c\mid\hole_{i,T}\mid (\vec D_\bullet) 
  \\
  v_\bullet \ &\mis c\mid\alpha_T\mid (\vec v_\bullet)
\end{align*}
Thus, $D_\bullet$ and $v_\bullet$ are values where functions are replaced by holes and abstract names,  respectively. To lighten notation, we denote them by $D,v$.

Moves $\eta$ are proponent call (\lpropapp{\alpha}{D}) and return (\lpropret{D}) moves involved in transitions from   
opponent to proponent configurations; and
opponent call (\lopapp{i}{v}) and return (\lopret{v}) moves in transitions from opponent to proponent configurations.

\begin{remark}
  Note the abstract names used in moves (and, later, traces) are of the form $\alpha_T$, i.e.\ without $i$-annotations. This amounts to the fact that any two abstract names $\alpha_T^{i},\alpha_T^{j}$  with $i\neq j$
  correspond to the same function played by opponent in two different points in the interaction. At each point, the proponent functions $\vec v$ that the opponent has access to may differ, and hence the need for different indices to distinguish the two instances of $\alpha_T$. 
  In the LTS, such distinction is not needed for proponent higher-order values as they are suppressed from proponent moves altogether.
\end{remark}

\begin{definition}[Traces]
We let a \emph{trace} $t$ be an alternating sequence of opponent/proponent moves. 
 We write $t+t'$ or, sometimes for brevity, $t\,t'$ to mean trace concatenation.
\end{definition}

\paragraph*{\flushleft Configurations:}
Proponent configurations are written as
$\pconf{A}{M}{K}{t}{e}{V}$ and proponent configurations as
 $\oconf{A}{M}{K}{t}{V}{\vec u}$.
All configurations are ranged over by $C$.
In these configurations:
\begin{itemize}
  \item $A$ is a partial map which assigns a sequence of names $\vec v$ to each abstract function name $\alpha$ (that has been used so far in the interaction) and integer index $j$.
  We write $\alpha^{j,\vec v}\in A$ for $A(\alpha, j)=\vec v$.
The index $j$ is used to distinguish between different uses of the same abstract function name $\alpha$ by opponent in the interaction.
    The sequence of values $\vec v$ represents the proponent functions that were available to opponent when the name $\alpha^j$ was used (knowledge accumulation for constructing context-generated functions, cf.~\cref{ex:fix-curried-ineq}).
    We write $A\uplus\alpha^{j,\vec v}$ for $A\cup((\alpha,j),\vec v)$, assuming $(\alpha,j)\not\in A$.
\item $t$ is the \emph{opponent-visible trace}, i.e.\ a subset of the current interaction that the opponent can have access to,
starting with a move where the proponent calls an opponent abstract function.
\item $K$ is a stack of proponent continuations, created by nested proponent calls.
  We call configurations with an empty stack \emph{top-level} and those with a non-empty stack \emph{inner-level}; opponent top-level configurations are also called \emph{final}.
  Configurations of the form $\pconf{\emptyA}{\emptyA}{\emptyK}{\emptytr}{e}{\emptyK}$ are called \emph{initial}.
\item $M$ is a set of opponent-visible traces. It ensures pure behaviour of the opponent (cf.~\cref{ex:pure}):
 it restricts the moves of the opponent when an opponent-visible trace is being run for a second (or subsequent) time.
Component $M$ is also examined by the bisimulation to determine equivalence of top-level configurations.
 It can be seen as a \emph{memory} of the behaviour of the opponent abstract functions so far and an oracle for future moves.
 Given $M$, we define a map from proponent-ending traces to next opponent moves:
 \[
\nextmove{M}{t} = \{ \eta \mid t\eta \in M\}
\]
We consider only \emph{legal} $M$'s whereby 
$|\nextmove{M}{t}|\leq1$ for any trace $t$ ending in a proponent move
and each abstract function name $\alpha$ appears at most once in $M$.
We write $M[t]$ for $M\cup\{t\}$.
We may also write $M_{C}$ for the $M$-component of a configuration $C$. 

\item $e$ is the proponent expression reduced in proponent configurations.
\item In opponent configurations, $\vec u$ is the sequence of values (proponent functions) that are available to opponent to call at the given point in the interaction.
In both kinds of configurations, $V$ is a stack of sequences of proponent functions. These components encode the opponent knowledge accumulation necessary for a sound NF bisimulation theory for \lang.
    They enable sequence of calls to proponent functions (cf.\ \cref{ex:u's,ex:v's}), and construction of opponent-generated abstract functions with the appropriate level of knowledge attached to them (cf.\ \cref{ex:fix-curried-ineq}).
\end{itemize}

\begin{figure*}[t] 

  \[\begin{array}{@{}cllll@{}}
    \irule[PropTau][proptau]{
      e \red e'
    }{
      \pconf{A}{M}{K}{t}{e}{V} \trans{\tau} \pconf{A}{M}{K}{t}{e'}{V}
    }
    \\[2em]
    \irule[PropRetBarb][propretbarb]{
      (D,\vec v) = \ulpatt(v)
    }{
      \pconf{A}{M}{\emptyK}{t}{v}{\vec u}
      \trans{\lpropret{D}}
      \oconf{A}{M}{\emptyK}{\emptytr}{\emptyK}{\vec v}
    }
    \\[2em]
    \irule[PropRet][propret]{
      (D,\vec v) = \ulpatt(v)
      \\
      K\not=\emptyK
      \\
      t' = t+\lpropret{D}
    }{
      \pconf{A}{M}{K}{t}{v}{\vec u,V}
      \trans{\tau}
      \oconf{A}{M[t']}{K}{t'}{V}{\vec u,\vec v}
    }
    \\[2em]
    \irule[PropCall][propappf]{
      (D,\vec v) = \ulpatt(v)
      \\
      t'= \lpropapp{\alpha}{D}
      \\
      \vec\alpha^{j,{\vec u}}\in A
    }{
      \pconf{A}{M}{K}{t}{E[\app {\alpha_{T_1 \arrow T_2}^{j}} v]}{V}
      \trans{\tau} 
      \oconf{A}{M[t']}{(t,E\hole_{T_2}),K}{t'}{V}{\vec u,\vec v}
    }
    \\[2em]
    \irule[OpRet][opret]{
      \nextmove{M}{t} \subseteq \{\lopret{D\hole[\vec \alpha]}\}
      \\
      (D,\vec \alpha) \in \ulpatt(T')
      \\
      t''=t+\lopret{D\hole[\vec \alpha]}
    }{
      \oconf{A}{M}{(t',E\hole_{T'},T),K}{t}{V}{\vec v}
      \trans{\tau}
      \pconf{A\uplus\vec\alpha^{j,{\vec v}}}{M[t'']}{K}
      {t'}{E\hole[D{\hole[\vec \alpha^{j}]}]}{V}
    }
    \\[2em]
    \irule[OpCallBarb][opappbarb]{
       v_i : T_1 \arrow T_2
       \\
        (D,\vec \alpha) \in \ulpatt(T_1)
      \\
      \vec\alpha\text{ fresh}
    }{
      \oconf{A}{M}{\emptyK}{\emptytr}{\emptyK}{\vec v}
      \trans{\lopapp{i}{D\hole[\vec \alpha]}}
      \pconf{A\uplus\vec\alpha^{0,{\emptytr}}}{M}{\emptyK}{\emptytr}{\app {v_i} {D\hole[\vec \alpha^{0}]}}{\emptytr}
    }
    \\[2em]
    \irule[OpCall][opappf]{
      \nextmove{M}{t} \subseteq \{\lopapp{i}{D\hole[\vec \alpha]}\}
      \\
      K\neq\emptyK
      \\
       v_i : T_1 \arrow T_2
       \\
       (D,\vec \alpha) \in \ulpatt(T_1)
      \\
        t'=t+\lopapp{i}{D\hole[\vec \alpha]}
    }{
      \oconf{A}{M}{K}{t}{V}{\vec v}
      \trans{\tau}
      \pconf{A\uplus\vec\alpha^{j,{\vec v}}}{M[t']}{K}{t'}{(\app {v_i} {D\hole[\vec \alpha^{j}]})}{\vec v,V}
    }

  \end{array}\]
  \hrule
  \caption{The Labelled Transition System. We denote by $\cdot$ the empty stack, and by $\varepsilon$ the empty sequence.}\label{fig:lts}
\end{figure*} 

\newcommand{\pconfb}[6]{\ma{\color{blue}\langle\color{black} #1 \mathop{;} #2 \mathop{;} {#3} \mathop{;} #4 \mathop{;} #5 \color{blue}\rangle\color{black}}}
\newcommand{\oconfb}[6]{\ma{\color{red}\langle\color{black} #1 \mathop{;} #2 \mathop{;} {#3} \mathop{;} #4 \mathop{;} #5 \color{red}\rangle\color{black}}}

\paragraph*{\flushleft Transitions:}
Transitions are of the form
$C \trans{l} C'$, where transition label $l$ is either an immediately observable move $\eta$ or a generic $\tau$, hiding any move involved in the transition. In the former case, observable moves can be opponent calls ($\OpApp$) or proponent returns ($\PropRet$).
Unlike standard LTSs, this LTS hides call/return moves involved in transitions of inner-level configurations, which are stored in the configuration memory $M$ instead.
As we will see later in this section, this is to allow equivalent terms to have different order of calls to opponent functions.
Only \emph{top-level} transitions contain move annotations, making them directly observable.
These are transitions produced by one of the barbed rules (\iref{propretbarb}, \iref{opappbarb}).
In the remaining transition rules moves are accumulated in traces which are stored in the memory component $M$ of the configurations. These will be examined by the bisimulation at top-level configurations.

The simplest transitions are those produced by the \iref{proptau} rule, embedding reductions into proponent configurations.
The remaining transitions involve interactions between opponent and proponent and are detailed below.

\paragraph*{\flushleft Proponent Return:}

When the proponent expression has been reduced to a value, the LTS performs a $\PropRet$-move, either by the \iref{propretbarb} transition, when the configuration is top-level, or the \iref{propret} transition, when it is not.
In both cases the value $v$ being returned is deconstructed to:
\begin{itemize}
\item an \emph{ultimate pattern} $D$ (cf.~\cite{LassenL07}), which is a context obtained from $v$ by replacing each function in $v$ with a distinct numbered hole; together with
\item a sequence of values $\vec v$ such that $D\hole[\vec v]=v$.
\end{itemize}
We let $\ulpatt(v)$ be a deterministic function performing this decomposition.

In rule \iref{propretbarb} the functions $\vec v$ obtained from $v$ become the knowledge of the resulting opponent configuration; opponent can call one of these functions to continue the interaction. 
The previous knowledge $\vec u$ stored in the one-frame stack is being dropped.
This dropping of knowledge is sufficient for a sound NF bisimulation theory based on this LTS, as justified by our soundness result and corroborated by the conditions of applicative bisimulation which encode top-level interactions without accumulating opponent knowledge from previous moves.

On the other hand, in \iref{propret}, $\vec v$ is added to the most current opponent knowledge $\vec u$, stored in the top-frame of the knowledge stack which is popped in the resulting configuration.
This is necessary because, in inner level configurations, opponent should be allowed to call a proponent function it knew before it called the function that returned $v$, allowing observations such as those in \cref{ex:u's,ex:v's}.

In \iref{propretbarb} the context $D$ extracted by ultimate pattern matching becomes observable in the transition label $\lpropret{D}$.
Again, this is in line with the definition of applicative bisimulation where the return values of top-level functions are observed by the bisimulation moves.
However, in rule \iref{propret} this observation is \emph{postponed}: the $\lpropret{D}$ move is appended to the current trace, and this trace is being stored in the $M$ memory in the configuration.
This memory will then be used to make distinctions between configurations in a bisimulation definition, when top-level transitions are reached.
This storing of inner-level moves makes unobservable the order and repetition of proponent calls to opponent functions in the LTS, allowing to prove equivalences such as the one in \cref{ex:pure}.

\paragraph*{\flushleft Proponent Call:}
Rule \iref{propappf} produces a transition when a call to an opponent abstract function ${\alpha_{T_1 \arrow T_2}^{j}}$ is at reduction position in a proponent expression.
Function $\ulpatt(v)$ is again used to decompose the call argument to context $D$ and functions $\vec v$,
whereas $\alpha^j$ is looked up in $A$ to identify the knowledge $\vec u$ attached to this use of the $\alpha$ name at the time $\alpha^j$ was created.
Then $\vec v$ and $\vec u$ are combined to create opponent's knowledge in the resulting configuration.
The trace $t$ accumulated in the (proponent) source configuration of the transition is being pushed onto the stack component $K$ together with the continuation of the expression being reduced.
This is because a proponent call transition triggers the creation of a new opponent-visible trace $t'$, starting with the call move.
This new trace is stored in the memory $M$ and used in the resulting (opponent) configuration.

Segmentation of traces into opponent-visible trace fragments, as performed by this rule, is important for full abstraction of the NF bisimulation defined below. 
When configurations are compared by the bisimulation, the exact interleaving of these trace segments is not observable as the language is pure and opponent-generated functions have only a local view of the overall computation.
Moreover, opponent-visible traces relate to O-views in game semantics but contain only a single (initiating) proponent call move.

\paragraph*{\flushleft Opponent Return:}
An opponent configuration with a non-empty stack component $K$ may return a value with rule \iref{opret}. In order to obtain this value
we extend $\ulpatt$ to the return type $T$ through the use of symbolic function names:
$\ulpatt(T)$ is the set of all pairs $(D,\vec \alpha_{\vec T})$ such that $\vdash D\hole[\vec\alpha] : T$, where $D$ is a value context that does not contain functions,
and the types of $\vec \alpha$ and the corresponding holes match. Note that in this definition we leave the $j$ annotation of $\alpha$'s blank as it is filled-in by the rule. 
In the resulting configuration $\alpha^{j,{\vec v}}$ is added in $A$, extending its domain by $(\alpha,j)$.

This transition can be performed in two cases; when:
\begin{itemize}
  \item $\nextmove{M}{t} = \emptyset$. In this case the current opponent-visible trace $t$ is not a strict prefix of a previously performed trace stored in $M$, and the configuration can non-deterministically perform this return transition.
    If it does, the resulting configuration stores in $M$ the extended trace $t''=t+\lopret{D\hole[\vec \alpha]}$. Note that $j$ is not stored in moves and thus neither in $M$.
    Moreover in this case the $\vec\alpha$ used are chosen fresh, this is guaranteed by the implicit condition that $M$ is legal and thus $\alpha$ cannot appear twice in $M$.

  \item $\nextmove{M}{t} = \{\lopret{D\hole[\vec \alpha]}\}$. In this case the current configuration is along an opponent-visible trace that has occurred previously and performed a return as a next move. Thus because the opponent must have purely functional behaviour, the configuration can perform no other but this return transition.
\end{itemize}
If $\nextmove{M}{t}$ does not fall into one of the above cases the transition does not apply.

To encode functional behaviour, the current opponent knowledge $\vec v$ can only be stored in the abstract functions $\vec \alpha$ generated at this transition and stored in $A$.
It \emph{cannot} be carried forward otherwise in the resulting proponent function. Hence, if $T'$ is a base type, this knowledge is lost after the transition.

\paragraph*{\flushleft Opponent Call:}

The proponent function being called in these transitions defined by \iref{opappbarb} and \iref{opappf} is one of those in the current opponent knowledge $\vec v$. We use the relative index $i$ in $\vec v$ to refer to the function being called.
The argument supplied to this function is obtained again by the function $\ulpatt$ applied to the argument type $T_1$.

Opponent call transitions are differentiated based on whether they are top- or inner-level.
Top-level opponent calls (\iref{opappbarb}) are immediately observable and thus transitions are annotated with the move.
Moreover, the opponent knowledge is dropped at the transition and not accumulated in the knowledge stack or created abstract function names.
This is in line with applicative bisimulation where related top-level functions are called only at the point they become available in the bisimulation, and are provided with identical arguments, thus not not containing any related functions from the observer knowledge.

However inner-level opponent calls are not immediately observable and thus the corresponding move is stored in traces in $M$.
As for inner opponent return transitions, $\nextmove{M}{t}$ may require that the transition must or cannot be applied.

\paragraph*{\flushleft Big-Step bisimulation:}

\begin{definition}[Trace transitions]
We use $\xtoo{}$ for the reflexive and transitive closure of the $\trans{\tau}$ transition.
We write $C\xtoo{\eta}C'$ to mean $C\xtoo{}\trans{\eta}\xtoo{}C'$;
and $C\xtoo{t}C'$ to mean $C\xtoo{\eta}\xtoo{t'}C'$ when $t=\eta t'$, and $C\xtoo{}C'$ when $t$ is empty.
\end{definition}

\noindent
Note that, by definition, trace transitions derived by our LTS only contain $\OpApp$
and $\PropRet$ moves.

\begin{definition}
  Given a closed expression $\vdash e:T$, the initial configuration associated to $e$ is:
  \[
C_e = \pconf{\emptyA}{\emptyA}{\emptyK}{\emptytr}{e}{\emptyK}
\]
Accordingly, we can give the semantics of $e$ as:
\[
\sem{e} = \{ (t,M)\mid \exists A,t',V,\vec v.\  C_e\xtoo{t}\oconf{A}{M}{\emptyK}{t'}{V}{\vec v}\}.
\]
\end{definition} 

A closed expression $e$ will be first evaluated by the LTS using the operational semantics rules (and \textsc{PropTau}). Once a value is reached, this will be communicated to the context by means of a proponent return (rule \textsc{PropRetBarb}), after it has been appropriately decomposed. For there on, the game continues with opponent interrogating functions produced by proponent (using rule \textsc{OpAppBarb}). Proponent can interrogate functions provided by opponent (\textsc{PropApp}), leading to further interaction all of which remains hidden (see $\tau$-transitions), until proponent provides a return to opponent's top-level application (\textsc{PropRetBarb}).

\begin{example}\label{ex:pure-trans}
We now revisit the terms in \cref{ex:pure} to show how our LTS works.
  We start with term $M_1$ from \cref{ex:pure} placed in an initial configuration $C_1=\pconf{\emptyA}{\emptyA}{\emptyK}{\emptytr}{M_1}{\emptyK}$.
  The first is a proponent return transition which moves the function into the opponent's knowledge.
  \[
    C_1 \trans{\lpropret{\hole[]}} \oconf{\emptyA}{\emptyA}{\emptyK}{\emptytr}{\noe}{M_1}= C_{12}
  \]
  Then opponent calls and proponent immediately returns the second function (\lstinline{fun g -> ...}), which we call $M_{11}$, and opponent calls $M_{11}$; all are top-level interactions.
  \begin{align*}
    C_{11} &\trans{\lopapp{1}{\alpha_f}}
            \pconf{\alpha_f^0}{\emptyA}{\emptyK}{\emptytr}{M_{11}\sub{\text{\lstinline{f}}}{\alpha_f^0}}{\emptyK}\\
           &\trans{\lpropret{\hole[]}} 
            \oconf{\alpha_f^0}{\emptyA}{\emptyK}{\emptytr}{\noe}{M_{11}\sub{\text{\lstinline{f}}}{\alpha_f^0}}\\
           &\trans{\lopapp{1}{\alpha_g}} 
            \pconf{\alpha_f^0,\alpha_g^0}{\emptyA}{\emptyK}{\emptytr}{M_{12}}{\emptyK} = C_{12}
  \end{align*}
  where

\medskip
\noindent
\begin{tabular}{@{}l@{\;}l@{}}
  $M_{12} =$
&\begin{minipage}[t]{0.8\columnwidth}\vspace{-1.3em} \begin{lstlisting}[boxpos=t]
if $\alpha_f^0$ () == $\alpha_g^0$ () then
  if $\alpha_f^0$ () == $\alpha_g^0$ () then 0 else 1
else 2
\end{lstlisting}\end{minipage}
\end{tabular}\\
  The following transition is a proponent call of $\alpha_f^0$, followed (necessarily, due to types) by an opponent return.
  \begin{align*}
    C_{12} &\trans{\tau}
            \oconf{\alpha_f^0,\alpha_g^0}{\{t_1\}}{(\emptytr,E_1)}{t_1}{\emptyK}{\emptyK}
            \tag{inner-level move \lpropapp{\alpha_f}{()}, rule \iref{propappf}}\\
           &\trans{\tau}
            \pconf{\alpha_f^0,\alpha_g^0}{\{t_2\}}{\emptyK}{\emptytr}{E_1[k_1]}{\emptyK} = C_{13}
            \tag{inner-level move \lopret{k_1}, rule \iref{opret}}
  \end{align*}
  where $t_1 = \lpropapp{\alpha_f}{()}$ and $t_2 = t_1,\lopret{k_1}$ and $k_1$ is an integer constant and 
  $E_1=$ (\lstinline{if $\hole{}_\Int$ == $\alpha_g^0$ () then ...}).
  The transitions continue with the call and return of $\alpha_g$.
  \begin{align*}
    C_{13} &\trans{\tau}
            \trans{\tau}
            \pconf{\alpha_f^0,\alpha_g^0}{M_1}{\emptytr}{\emptyK}{E_2[k_2]}{\emptyK} = C_{14}
            \tag{inner-level moves \lpropapp{\alpha_g}{()} and then \lopret{k_2}}
  \end{align*}
  where $M_1=\{t_2,t_3\}$ and $t_3 = \lpropapp{\alpha_f}{()},\lopret{k_1}$ and $k_2$ is an integer constant and 
  $E_2=$ (\lstinline{if $k_1$ == $\hole{}_\Int$ then ...}).
  
  The behaviour of $\alpha_f$ and $\alpha_g$ are now determined at this point from the traces $t_2$ and $t_3$ in the memory component of the configuration.
  Thus the following transitions only depend on whether $k_1=k_2$. If they are not equal, proponent returns with a single
  $\lpropret{2}$ transition.
  \begin{align*}
    C_{14} & \trans{\lpropret{2}}
            \oconf{\alpha_f^0,\alpha_g^0}{M_1}{\emptyK}{\emptytr}{\emptyK}{\emptyK} = C_{15}(k_1,k_2)
            \tag{with $k_1\not=k_2$ in $M_1$}
  \end{align*}

  If they are equal, the remaining transitions will be the following ones, reaching final configuration $C_{15}$.
  \begin{align*}
    C_{14} &\trans{\tau}\trans{\tau}\trans{\tau}\trans{\tau}\hfill\;
    \tag*{(inner-level moves \nbox{\lpropapp{\alpha_f}{()}, \lopret{k_1},\\\lpropapp{\alpha_g}{()},\lopret{k_2})}}\\
           &\trans{\lpropret{0}}
            \oconf{\alpha_f^0,\alpha_g^0}{M_1}{\emptyK}{\emptytr}{\emptyK}{\emptyK} = C_{15}(k_1,k_2)
            \tag{with $k_1=k_2$ in $M_1$}
  \end{align*}
  The $\lopret{k_1}$ and $\lopret{k_2}$ moves are the only possible at those points of the trace due to the memory component $M_1$, encoding a purely functional behaviour of the opponent.

  Therefore, $C_1$ has only the following trace transitions:
  \begin{align*}
    C_1 & \xtoo{\lpropret{\hole[]}} \xtoo{\lopapp{1}{\alpha_f}} \xtoo{\lpropret{\hole[]}} \xtoo{\lopapp{1}{\alpha_g}} \xtoo{\lpropret{2}} C_{15}(k_1,k_2)
            \tag{with $k_1\not=k_2$}\\
    C_1 & \xtoo{\lpropret{\hole[]}} \xtoo{\lopapp{1}{\alpha_f}} \xtoo{\lpropret{\hole[]}} \xtoo{\lopapp{1}{\alpha_g}} \xtoo{\lpropret{0}} C_{15}(k_1,k_2)
            \tag{with $k_1=k_2$}
  \end{align*}

  Configuration $C_1'=\pconf{\emptyA}{\emptyA}{\emptyK}{\emptytr}{N_1}{\emptyK}$ has the same trace transitions, with fewer inner call and return moves, in different order, resulting in top-level configurations with the same memory as the corresponding ones above.
  One can thus see that the two original terms $M_1$ and $N_1$ are equivalent under this LTS as $\sem{M_1}=\sem{N_1}$.
\end{example}

\begin{example}
  Here we explore the inequivalent terms from \cref{ex:v's} to show how they are differentiated by our LTS. We focus
  on the configuration $C_2= \pconf{\emptyA}{\emptyA}{\emptyK}{\emptytr}{M_2}{\emptyK}$ and the transitions that differentiate it from $N_2$, corresponding to the behaviour of the context shown in \cref{ex:v's}.

  To simplify notation in this example, we identify memory components with the same prefix closure and use in configurations below an instructive representative of each memory component equivalent class.

  Since $M_2$ is a curried two-argument function similar to $M_1$ in the preceding example, we will have the following initial transitions.
\begin{align*}
    C_{2}  &
            \trans{\lpropret{\hole[]}}
            \trans{\lopapp{1}{\alpha_f}}
            \trans{\lpropret{\hole[]}} \\
            &\trans{\lopapp{1}{\true}} \xtoo{}
            \pconf{\alpha_f^0}{\emptyA}{\emptyK}{\emptytr}{M_{21}}{\emptyK} = C_{21}
  \end{align*}
  where $M_{21}$ $=$ \lstinline{$\alpha_f^0$ (X, fun _ ->  tt)}
  and \lstinline{X} is the recursive function \lstinline{fix X(d) -> $\alpha_f^0$ (X, fun _ -> d)}.
The next transition is an inner proponent call.
\begin{align*}
    C_{21}  &
            \trans{\tau}
            \oconf{\alpha_f^0}{\{t_1\}}{K_1}{t_1}{\emptyK}{v_1,v_2} = C_{22}
            \tag{inner move \lpropapp{\alpha_f}{(\hole[],\hole[])}}
  \end{align*}
  where $v_1=\text{\lstinline{X}}$ and $v_2=\text{\lstinline{fun _ -> tt}}$ and $t_1=\lpropapp{\alpha_f}{(\hole[],\hole[])}$
  and $K_1=(\emptytr, \hole)$.
  At this point of the interaction, opponent can either return a value or call one of the $v_i$. Since we are focusing on the behaviour of the discriminating context we show the following call to $v_2$:
\begin{align*}
    C_{22}  &
            \trans{\tau}\xtoo{}
            \pconf{\alpha_f^0}{\{t_2\}}{K_1}{t_2}{\true}{V_1} = C_{23}
            \tag{inner move \lopapp{2}{\false}}
  \end{align*}
  where the one-frame stack $V_1$ is $(v_1,v_2)$ and 
  $t_2=t_1,\lopapp{2}{\false}$.
  This is followed by
\begin{align*}
    C_{23}  &
            \trans{\tau}
            \oconf{\alpha_f^0}{\{t_3\}}{K_1}{t_3}{\emptyK}{v_1,v_2}
            \tag{inner \lpropret{\true}}\\
            &\trans{\tau}\xtoo{}
            \pconf{\alpha_f^0}{\{t_4\}}{K_1}{t_4}{M_{22}}{V_1} = C_{24}
            \tag{inner \lopapp{1}{\false}}
  \end{align*}
  where $t_3=t_2,\lpropret{\true}$
  and
  $t_4=t_3,\lopapp{1}{\false}$
  and
  $M_{22}$ $=$ \lstinline{$\alpha_f^0$ (X, fun _ ->  ff)}.
  Then the recursive call results to the transition:
\begin{align*}
    C_{24}  &
            \trans{\tau}
            \oconf{\alpha_f^0}{M_1}{K_2}{t_1}{V_1}{v_1,v_2'} = C_{25}
            \tag{inner \lpropapp{\alpha_f}{(\hole[],\hole[])}}
  \end{align*}
  where $M_1=\{t_4\}[t_1]=\{t_4\}$ and 
  $K_2 = (t_4,(\hole[],\hole[])),K_1$ and
  $v_2'=\text{\lstinline{fun _ -> ff}}$.
  At this point opponent must necessarily call $v_2'$ as the current trace $t_1$ in the configuration is a prefix of $t_4$ in the memory component $M_1$ and 
  $\nextmove{M}{t_1}=\lopapp{2}{\false}$:
\begin{align*}
    C_{25}  &
            \trans{\tau}\xtoo{}
            \pconf{\alpha_f^0}{M_2}{K_2}{t_2}{\false}{V_2}
            \tag{inner \lopapp{2}{\false}}\\
            &\trans{\tau}
            \oconf{\alpha_f^0}{M_3}{K_2}{t_3'}{V_1}{v_1,v_2'}
            = C_{26}
            \tag{inner \lpropret{\false}}
  \end{align*}
  where $t_3'=t_2,\lpropret{\false}$ and
  $M_2=\{t_4\}[t_2]=\{t_4\}$ and
  $M_3=\{t_4\}[t_3']=\{t_4,t_3'\}$ and
  $V_2=(v_1,v_2'),V_1$.
  Trace $t_3'$ is not a prefix of $t_4$ and therefore opponent can perform any move, including returning a value.
\begin{align*}
    C_{26}  &
            \trans{\tau}
            \pconf{\alpha_f^0}{M_4}{K_1}{t_4}{\true}{V_1}
            \tag{inner \lopret{\true}}\\
            &\trans{\tau}
            \oconf{\alpha_f^0}{M_5}{K_1}{t_5}{\emptyK}{v_1,v_2}
            \tag{inner \lpropret{\true}}\\
            &\trans{\tau}
            \pconf{\alpha_f^0}{M_6}{\emptyK}{\emptytr}{\true}{\emptyK}
            \tag{inner \lopret{\true}}\\
            &\trans{\lpropret{\true}}
            \oconf{\alpha_f^0}{M_6}{\emptyK}{\emptytr}{\emptyK}{\emptyK}
  \end{align*}
  where
  $M_4=M_3[t_4']=\{t_4,t_4'\}$ and
  $t_4'=t_3',\lopret{\true}$ and
  $M_5=\{t_5,t_4'\}$ and
  $t_5=t_4,\lpropret{\true}$ and
  $M_6=\{t_6,t_4'\}$ and
  $t_6=t_5,\lopret{\true}$.
  
  Term $N_2$ is not able to perform this transition trace as once the first function of the pair is called, it diverges.
\end{example}

 We note that the LTS is deterministic at proponent configurations, but not at opponent configurations as the latter can fire more than one $\tau$-transitions. Nonetheless, as the behaviour of opponent is restricted by the memory $M$, we can show the following.

\begin{lemma}[$M$-determinacy]\label{lem:Mdet}
  Given final configurations $C,C_1,C_2$ such that
  $C \xtoo{\lopapp{i}{D[\vec\alpha]}\lpropret{D'}} C_i$ (for $i=1,2$),
  if $M_{C_1}\cup M_{C_2}$ is legal then $C_1= C_2$.
\end{lemma}
\begin{proof}
  Let us set $M=M_{C_1}\cup M_{C_2}$. We break down the transitions from $C$ to $C_i$ as follows:
\[
  C \trans{\lopapp{i}{D[\vec\alpha]}}
C_{i0}\trans{\tau}C_{i1}\trans{\tau}\cdots\trans{\tau}C_{in_i}
  \xtoo{\lpropret{D'}} C_i
\]
We show that, for each $j=0,\dots,n$, where $n=\min(n_1,n_2)$, $C_{1j}= C_{2j}$. We do induction on $j$; the base case is clear from the LTS rules. Now consider $C_{ij}\trans{\tau}C_{i(j+1)}$. By IH,  $C_{1j}= C_{2j}$. If they are (both) $P$-configurations then, by determinacy of proponent, the two configurations are bound to make the same move. Hence, 
$C_{1(j+1)}= C_{2(j+1)}$. If $C_{ij}$ are $O$-configurations, let their current (common) trace component be $t$. As the memory maps $M_1,M_2$ of $C_{1(j+1)},C_{2(j+1)}$ are included in $M$, we must have $\nextmove{M_1}{t}=\nextmove{M_2}{t}$ and, hence, $C_{1j}$ and $C_{2j}$ must have made the same move and therefore $C_{1(j+1)}=C_{2(j+1)}$. Now, since $C_{1n}=C_{2n}$ and one of them makes a $P$-return then, by determinacy of proponent, the other must make the same  $P$-return. Hence, $C_1=C_2$.
\end{proof}

Weak bisimulation is defined in the standard way, albeit using the big-step transition relation corresponding to initial and final configurations. In addition, $M$-components are compared \emph{contravariantly}: $M$ records the opponent behaviour faced by the simulated expression, which then restricts that faced by the simulating expression..

\begin{definition}[Weak (Bi-)Simulation]
A  binary relation \bisim{R} on initial and final configurations is a \emph{weak simulation} when
  for all $C_1 \bisim*{R} C_2$:
  \begin{itemize}
    \item \emph{Initial configurations:} if $C_1 \xtoo{\lpropret{D'}} C_1'$, there exists $C_2'$ such that $C_2 \xtoo{\lpropret{D'}} C_2'$ and $C_1' \bisim*{R} C_2'$;
    \item \emph{Final configurations:} if $C_1 \xtoo{\lopapp{i}{D[\vec\alpha]}\lpropret{D'}} C_1'$ with $\vec\alpha$ fresh for $C_2$, there exists $C_2'$ such that $C_2 \xtoo{\lopapp{i}{D[\vec\alpha]}\lpropret{D'}} C_2'$ and $C_1' \bisim*{R} C_2'$;
  \item $M_{C_2}\subseteq M_{C_1}$ (where $M_{C_i}$ is the $M$-component of $C_i$).
\end{itemize}
  If \bisim{R}, $\bisim{R}^{-1}$ are weak simulations then
  \bisim{R} is a \emph{weak bisimulation}.
  Similarity $(\simil)$ and bisimilarity $(\bisimil)$ are the largest weak 
  simulation and bisimulation, respectively.
  \defqed
\end{definition}

This definition resembles that of applicative bisimulation for \lang, in that related top-level functions applied to identical arguments must co-terminate and return related results.
However the most important difference here is that there is no quantification over all possible programs. The context $D$ is a value without any functions in it (essentially containing constants and/or pairs) which is determined by the type of the $i$'th function. The fresh names $\vec\alpha$ correspond to opponent-generated functions but are first-order entities that are equivalent up to renaming. Thus this definition constitutes a big-step Normal Form bisimulation.



\begin{definition}[Bisimilar Expressions]
Expressions $\vdash e_1: T$ and $\vdash e_2: T$ are bisimilar, written $e_1 \bisimil e_2$, when 
  $C_{e_1} \bisimil C_{e_2}$. 
  \defqed
\end{definition}

\begin{lemma}\label{lem:somelemma}
$e_1 \bisimil e_2$ iff $\sem{e_1}=\sem{e_2}$.
\end{lemma}
\begin{proof}
  Note first that if $e_1 \bisimil e_2$ and $(t,M)\in\sem{e_1}$ then, starting from $C_{e_2}$, we can simulate the transitions producing $t$ and arrive at the same $M$. Conversely, suppose that $\sem{e_1}=\sem{e_2}$ and define:
  \[
    \bisim{R} = \{ (C_1,C_2)\mid
M_{C_1}= M_{C_2}\land
    \exists t.\,C_{e_i}\xtoo{t} C_i \land \text{$C_i$ final}\}.
  \]
We show that $\bisim{R}$ is a weak bisimulation.
Suppose $C_1\bisim*{R}C_2$ with trace $C_{e_i}\xtoo{t}C_i$,
and let $C_1 \xtoo{\lopapp{i}{D[\vec\alpha]}\lpropret{D'}} C_1'$ with $\vec\alpha$ fresh for $C_2$.
As $\sem{e_1}=\sem{e_2}$, there is a transition sequence:
\[
C_{e_2}\xtoo{t}\hat C_2\xtoo{\lopapp{i}{D[\vec\alpha]}\lpropret{D'}} C_2'
\]
such that $M_{C_1'}=M_{C_2'}$. Since $M_{C_2}=M_{C_1}\subseteq M_{C_1'}$, we have $M_{C_2}\subseteq M_{C_2'}$. Hence, starting from $C_{e_2}$ and repeatedly applying Lemma~\ref{lem:Mdet}, we conclude that $C_2=\hat C_2$, and thus $C_2$ can match the challenge of $C_1$. Hence, $\bisim{R}$ is a weak simulation and, by symmetry, a weak bisimulation.
\end{proof}

The previous result can be used to show that 
bisimilarity is sound and complete with respect to contextual equivalence. The proof is discussed in the next section.

\begin{theorem}[Full abstraction]\label{theorem:SC}
  $e_1 \bisimil e_2$ iff $e_1 \cxteq e_2$.  
\end{theorem}

\begin{remark}
Following~\cite{ChurchillEtal2020}, we can define $\sem{e_1}\leq\sem{e_2}$ to hold if
$\forall (t,M_1)\in\sem{e_1}.\,\exists M_2.\, (t,M_2)\in\sem{e_2} \wedge M_2\subseteq M_1$. Then, \Cref{lem:somelemma} can be sharpened to its similarity variant, which would lead to full abstraction of normal-form similarity.
\end{remark}
%
%
%


  \section{Full Abstraction}
  \label{sec:fa}
  To prove that the LTS is sound and complete we use an extended LTS based on operational game semantics~\cite{Laird07}. The latter differs from our main LTS in that proponent and opponent can play the same kinds of moves, and in particular they can pass fresh function names to the other player, or apply functions of the other player by referring to their corresponding names. This duality in roles allows for the modelling of both expressions and contexts.
Moreover, all moves are recorded in the trace, not just top-level ones, which in turn enables us to compose two LTS's corresponding respectively to an expression and its context.

We shall call this the \emph{game-LTS}, whereas the main LTS shall simply be \emph{the/our LTS}. We shall be re-using some of our main LTS terminology here, for example traces will again be sequences of moves, albeit of different kind of moves. This is done for notional economy and we hope it is not confusing.

\subsection{The game-LTS}

We start by introducing an enriched notion of trace. Traces shall now consist of \emph{moves} of the form:
\begin{align*}
\text{Moves} &&  m \ &::= \ p\mid o\\
\text{Proponent moves} &&  p \ &::= \ \Epropapp{\al}{D}{\vec\be}\mid \Epropret{D}[\vec\be]\\
\text{Opponent moves} &&  o \ &::= \ \Eopapp{\be}{D}{\vec\al}\mid \Eopret{D}[\vec\al]
\end{align*}
where $\al,\be$ (and variants thereof) are sourced from disjoint sets
$\ONs$ and $\PNs$  
of \emph{opponent} and \emph{proponent names} respectively.
Names represent abstract functions and
are used to abstract away the functions that a context and an expression are producing in a computation.
We shall often be abbreviating ``proponent'' and ``opponent'' to $P$ and $O$ respectively and write, for instance, ``$O$-moves'' or ``$P$-names''.

A \boldemph{complete trace} is then given by the following grammar.
\begin{align*}
  CT \ &\to \ CT_P\mid CT_O\\
  CT_P \ &\to \ \Epropret{D}[\vec\be]\ CT_{OP}\\
  CT_O \ &\to \ \Eopret{D}[\vec\al]\ CT_{PO}\\
  CT_{OP} \ &\to \ \emptytrace\mid \Eopapp{\be}{D}{\vec\al}\ CT_{PO}\, \Epropret{D}[\vec\be]\ CT_{OP}\\
  CT_{PO} \ &\to \ \emptytrace\mid \Epropapp{\al}{D}{\vec\be}\ CT_{OP}\, \Eopret{D}[\vec\al]\ CT_{PO}
\end{align*}
A \emph{trace} is a prefix of a complete trace.
  A trace $t$ is called \emph{legal} if it satisfies these conditions:
  \begin{itemize}
\item
  for each $t'p\sqsubseteq t$ with  $p=\Epropapp{\al}{D}{\vec\be}$ or $p=\Epropret{D}{\vec\be}$:
  \begin{itemize}
    \item 
   $\vec\be$ do not appear in $t'$\,---\,we say that move $p$ \emph{introduces} each $\be\in\vec\be_i$\,---\,and
 \item
   there is some move $o'$ in ${t'}$ that introduces $\al$;
 \end{itemize}
  \item for each $t'o\sqsubseteq t$ with $o=\Eopapp{\be}{D}{\vec\al}$ or $o=\Eopret{D}{\vec\al}$:
    \begin{itemize}
    \item $\vec\al$ do not appear in $t'$\,---\,we say that move $o$ \emph{introduces} each $\al_i\in\vec\al$\,---\,and
      \item 
  there is some move $p'$ in ${t'}$ that introduces $\be$.
  \end{itemize}
 \end{itemize}
 Thus, in a legal trace all applications refer to names introduced earlier in the trace. Put otherwise, all function calls must be to functions that are available when said calls are made.
 We say that an application
$\Eopapp{\be}{D}{\vec\al}$ (or $\Epropapp{\al}{D}{\vec\be}$)
is \emph{justified} by the (unique) earlier move that introduced $\be$ (resp.\ $\al$).
On the other hand, a return is justified by the call to which it returns.
In a legal trace, all call moves are justified.

Due to the modelled language being functional, not all names are visible to the players (i.e.\ proponent and opponent) at all times. For example if  opponent makes two calls to proponent function $\be$, say first $\Eopapp{\be}{D_1}{\vec\al_1}$ and later $\Eopapp{\be}{D_2}{\vec\al_2}$, the second call will hide from proponent all the trace related to the first one.
This limitation is captured by the notion of \emph{view}.
 Given a legal trace $t$, we define its \emph{$P$-view} $\pview{t}$ and \emph{O-view} $\oview{t}$ respectively as follows:
\begin{align*}
  \pview{t} &=\begin{cases}
    t & \text{if }|t|\leq1\\
    \pview{t'}p & \text{if }t=t'p\\
    \pview{t'}po & \text{if $t=t'pt''o$ and $o$ is justified by $p$} 
    \end{cases}
  \\
  \oview{t} &=\begin{cases}
    t & \text{if }|t|\leq1\\
    \oview{t'}o & \text{if }t=t'o\\
    \oview{t'}op & \text{if $t=t'ot''p$ and $p$ is justified by $o$}
    \end{cases}
\end{align*}
  We will focus on traces where each player's moves are uniquely determined by their current view. This corresponds to game-semantics \emph{innocence} (cf.\ \cite{HO}).

In the following definitions
  we employ basic elements from nominal set theory~\cite{nom2} to formally account for names in our constructions. Let us write $\mathcal{N}$ for $\ONs\uplus\PNs$. Finite-support name permutations that respect $O$- and $P$-ownership of names are given by:
  \begin{align*}
    \Perm = \{\pi:\mathcal{N}\xrightarrow{\cong}\mathcal{N}\mid \exists X\subseteq_{\text{finite}}\mathcal{N}.\, \forall  y\in\mathcal{N}\setminus X.\,\pi(y)=y\\ 
{}\land\forall x\in X.\, x\in\ONs\iff\pi(x)\in\ONs \}
  \end{align*}
  Given a trace $t$ and 
  a permutation $\pi$, 
  we write $\pi\cdot t$ for the trace we obtain by applying $\pi$ to all names in $t$. We write $t\sim t'$ if there exists some $\pi$ such that $t'=\pi\cdot t$. The latter defines an equivalence relation, the classes of which we denote by $[t]$: 
  \[
[t] = \{\pi\cdot t\mid \pi\in\Perm\}.
    \] 
  Moreover, we define the sets of $O$-views and $P$-views of $t$ (under permutation) as:
  \begin{align*}
\PVs(t) &= \{ \pi\cdot\pview{t'}\mid t'\sqsubseteq t\land\pi\in\Perm\} \\
\OVs(t) &= \{ \pi\cdot\oview{t'}\mid t'\sqsubseteq t\land\pi\in\Perm\}
  \end{align*}
  \begin{definition}\label{def:plays}\normalfont
    A legal trace $t$ is called a \boldemph{play} if: 
    \begin{itemize}
      \item
 for each $t'p,t''o\sqsubseteq t$, the justifier of $p$ (of $o$) is included in $\pview{t'}$ (resp.\ $\oview{t''}$);
\item for all $t_1p_1,t_2p_2,t_1'o_1,t_2'o_2\sqsubseteq t$,
  \begin{itemize}
  \item if $\pview{t_1}\sim \pview{t_2}$ then $\pview{t_1p_1}\sim \pview{t_2p_2}$,
  \item if $\oview{t_1'}\sim \oview{t_2'}$ then $\oview{t_1'o_1}\sim \oview{t_2'o_2}$.
  \end{itemize}\end{itemize}
We refer to the conditions above as \emph{visibility} and \emph{innocence} respectively.
\end{definition}
Visibility and innocence are standard game conditions (cf.~\cite{HO,Mccusker}):
the former corresponds to the fact that an expression (or context) can only call functions in its syntactic context; while the latter enforces purely functional behaviour.
    

\begin{figure*}[t] 

  \[\begin{array}{@{}cllll@{}}
    \irule[PropRet][propret]{
      (D,\vec v) \in \ulpatt(v)
      \\
      t' = t+\Epropret{D}[\vec\be]
      \\
      \conc'=\conc\uplus[\vec\be\mapsto\vec v^{\vec\al}]
    }{
      \pconf{\A}{\conc}{K}{t}{v}{\vec\be',V}[\vec\al]
      \trans{\Epropret{D}[\vec\be]}
      \oconf{\A}{\conc'}{K}{t'}{V}{\vec\be',\vec\be}
    }
    \\[2em]
    \irule[PropTau][proptau]{
      e \red e'
    }{
      \pconf{\A}{\conc}{K}{t}{e}{V}[\vec\al] \trans{\tau} \pconf{\A}{\conc}{K}{t}{e'}{V}[\vec\al]
    }
    \\[2em]
    \irule[PropApp][propappf]{
      (D,\vec v) \in \ulpatt(v)
      \\
      t'= t+\Epropapp{\al}{D}{\vec\be}
      \\
      \A(\al) = \vec\be'
      \\
      \conc' = \conc\uplus[\vec\be\mapsto \vec v^{\vec\al}]
    }{
      \pconf{\A}{\conc}{K}{t}{E[\app {\al_{T_1 \arrow T_2}} v]}{V}[\vec\al]
      \trans{\Epropapp{\al}{D}{\vec\be}} 
      \oconf{\A}{\conc'}{(E\hole_{T_2},\vec\al),K}{t'}{V}{\vec \be',\vec \be}
    }
    \\[2em]
    \irule[OpRet][opret]{
      \nextmove{O}{t}\subseteq_\star[\Eopret{D}[\vec \al]] 
      \\
      (D,\vec \al) \in \ulpatt(T')
      \\
      t'=t+\Eopret{D}[\vec \al]
    }{
      \oconf{\A}{\conc}{(E\hole_{T'},\vec\al'),K}{t}{V}{\vec\be}
      \trans{\Eopret{D}[\vec\al]}
      \pconf{\A\uplus\vec\al^{\vec\be}}{\conc}{K}
      {t'}{E\hole[D{\hole[\vec \al]}]}{V}[\vec\al',\vec\al]
    }
    \\[2em]
    \irule[OpApp][opappf]{
      \nextmove{O}{t}\subseteq_\star [\Eopapp{\be_i}{D}{\vec \al}]
      \\
       \be_i : T_1 \arrow T_2
      \\
      \conc(\be_i)=v^{\vec\al'}
      \\
        (D,\vec \al) \in \ulpatt(T_1)
       \\
        t'=t+\Eopapp{\be_i}{D}{\vec \al}
    }{
      \oconf{\A}{\conc}{K}{t}{V}{\vec\be}
      \trans{\Eopapp{\be_i}{D}{\vec\al}}
      \pconf{\A\uplus\vec\al^{{\vec \be}}}{\conc}{K}{t'}{e}{\vec \be,V}[\vec\al',\vec\al]
    }
  \end{array}\]
  \hrule
  \caption{The Game Labelled Transition System (game-LTS).}\label{fig:lts3}
\end{figure*}

We can now proceed to the definition of the game-LTS. Similarly to the previous section,
we extend the language syntax of \cref{fig:lang} by including  O-names as values.
We define proponent and opponent \emph{game-configurations} respectively by:
  \[
\pconf{\A}{\conc}{K}{t}{e}{V}[\vec\al]
\quad\text{and}\quad
\oconf{\A}{\conc}{K}{t}{V}{\vec \be}
\]
and range over them by $\C$ and variants.
Here:
\begin{itemize}
\item $\A$ is a map which assigns to each (introduced) opponent name a sequence of proponent names.
  We write $\al^{\vec\be}\in \A$ for $\A(\al)=\vec\be$.
The sequence  $\vec\be$ are the proponent (function) names that were available to opponent when the name $\al$ was introduced.
\item Dually, $\kappa$ is a \emph{concretion map} which assigns to each (introduced) proponent name the function that it represents and the opponent names that are available to it.
  \item $t$ is a play recording all the moves that have been played thus far. Given $t$, we define
the partial function $\nextmove{O}{t}$, which we use to impose innocence on $O$-moves, by:
\[
  \nextmove{O}{t} = \{ \pi\cdot o \mid \exists t'o\sqsubseteq t.\, \oview{t}=\pi\cdot \oview{t'}\land t(\pi\cdot o) \text{ a play}\}
\]
When we write $\nextmove{O}{t}\subseteq_\star[o]$, for some $o,t$, we mean that either $o\in\nextmove{O}{t}$ or $\nextmove{O}{t}=\emptyset$.
\item $K$ is a stack of proponent continuations (pairs of evaluation contexts and opponent names $\vec\al$), and $e$ is the expression reduced in proponent configurations. 
    \item $\vec\al$ and $\vec\be$ are sequences of other-player names that are available to proponent and opponent respectively at the given point in the interaction; $V$ is a stack of $\vec\be$'s.
\end{itemize}
Note that we store the full trace in configurations and we use names ($\be$ and variants) to abstract proponent higher-order values. There is no need of an $M$-component as we can rely on the full play.
We call a configuration
   \emph{initial} if it is in one of these forms (called respectively \emph{$P$- and $O$-initial}):\footnote{we write $V=\emptyK$ for an empty stack, and $V=\emptyseq$ for a singleton stack containing the empty sequence; moreover, here and elsewhere, we use underscore ($\_$) to denote any component of the appropriate type.}
  \[
  \C_e =  \pconf{\emptyA}{\cdot}{\emptyK}{\emptytrace}{e}{\emptyseq}[\emptyseq]
    \,\text{ or }\,
\C_{E} = \oconf{\emptyA}{\cdot}{(E\hole_T{:}\Unit,\emptyseq)}{\emptytrace}{\emptyK}{\emptyseq}
\]
and
   \emph{final} if it is in one of these forms  (\emph{$O$- and resp.\ $P$-final}):
   \[
     \oconf{\_}{\_}{\emptyK}{\_}{\emptyK}{\_}
    \quad\text{or}\quad
    \pconf{\_}{\_}{\emptyK}{\_}{()}{\emptyK}[\_].
\]
Note that, by definition of the LTS, a $P$-initial configuration can only lead to $O$-final configurations, whereas 
$O$-initial configurations lead to $P$-final configurations.

\begin{definition}
The game-LTS is defined by the rules in \cref{fig:lts3}.
Given initial configuration $\C$, 
we set:
\[
\CP(\C) = \{ t\in\Tr(\C)\mid t\text{ complete}\}
\]
where we let $\Tr(\C)$ be the set of plays produced by the LTS starting from $\C$.
\end{definition}

We can show that the traces produced by the game-LTS are plays and define a model for \lang based on sets of complete plays, but that would not be fully abstract.
Though presented  in operational form, our game-LTS is equivalent to the (base) game-model of \lang~\cite{PCFv}. Consequently,
if we model expressions by the sets of complete plays they produce, we miss even simple equivalences like $\lambda f.\,f()\equiv\lambda f.\,f(f())$\,---\,plays are too intentional and do not take into account the limitations of functional contexts. To address this, one can use a semantic quotient (cf.~\cite{PCFv}) or, alternatively, group the plays of an expression into sets of plays so as to profile functional contexts the expression may interact with (cf.~\cite{ChurchillEtal2020}). Thus, an expression is modelled by a \emph{set of sets of plays}, one for each possible context.
We follow the latter approach, and also combine it with the fact that {applicative tests suffice} (cf.~\cref{prop:CL}). 

\begin{definition}
Given a $P$-starting play $t$, we call a move $m$ of $t$ \boldemph{top-level} if:
\begin{itemize}
\item either $m$ is the initial $P$-return of $t$,
\item or $m$ is an $O$-call justified by a top-level $P$-move,
\item or $m$ is a $P$-return to  a top-level $O$-move.
\end{itemize}
We say that $t$ is \boldemph{top-linear} if each top-level $O$-move in $t$ is justified by the $P$-move that precedes it.
\end{definition}

Hence, top-level moves are those that start from or go to a final configuration.
If $t$ is complete and top-linear then:
\[
t = p_0o_1\cdots p_1\cdots o_n\cdots p_n\quad\text{and}\quad
\oview{t} = p_0o_1p_1\cdots o_n p_n
\]
where each $o_{i+1}$ is justified by $p_i$, each $p_i$ returns $o_i$ ($i>0$), and the $o_i,p_i$ above are all the top-level moves in $t$.
This means that, at the top level of a top-linear play, opponent may only choose one of the functions provided by proponent in their last move and examine it (i.e.\ call it), which precisely corresponds to what an applicative context would be able to do.

We can now present our main results for the game-LTS.
Given initial $P$-configuration $\C$, we define:
\begin{align*}
\OVs(\C) &= \{ \OVs(t) \mid t\in\CP(\C)\}\\
\OVTLs(\C) &= \{ \OVs(t) \mid t\in\CP(\C)\text{ and $t$ top-linear}\}
\end{align*}

  \begin{proposition}[Correspondence]
    Given $\vdash e_1,e_2:T$, $\OVTLs(\C_{e_1})=\OVTLs(\C_{e_2})$ iff $\sem{e_1}=\sem{e_2}$.
  \end{proposition}

\begin{proposition}[Game-LTS full abstraction]
Given $\vdash e_1,e_2:T$,    $e_1\cxteq e_2$ iff $\OVTLs(\C_{e_1})=\OVTLs(\C_{e_2})$. 
\end{proposition}

\cref{theorem:SC} follows from the two results above.
For the first result we build a translation from the game-LTS to the (plain) LTS that forms a certain bisimulation between the two systems.
To prove full abstraction of the game-LTS we use standard and operational game semantics techniques (cf.~\cite{HO,Laird07,TzeGhica}) along with the characterisation of PCF equivalence by sets of $O$-views presented in~\cite{ChurchillEtal2020}. 


  \section{Prototype Implementation}
  \label{sec:tool}
  We implemented the LTS with symbolic higher-order transitions in a prototype bisimulation checking tool for programs written in an ML-like syntax for \lang. Our tool implements a Bounded Symbolic Execution--via calls to Z3--for a big-step bisimulation of the LTS; the tool was developed in OCaml\footnote{\url{https://github.com/LaifsV1/pcfeq}}.

The tool performs symbolic execution of base type values through an extension of the LTS to include a \textit{symbolic environment} $\sigma : \Val \to \Val$ that accumulates constraints on \textit{symbolic constants} $\varkappa \in \Val$ that extend the set of values. Symbolic constants are of base type and may only be introduced by opponent moves (arguments and return values) and by reducing expressions that involve symbolic constants; their semantics follows standard symbolic execution. 
The exploration is performed over \textit{configuration pairs} $\langle C_1, C_2, M, \sigma, k \rangle$ of bisimilar term configurations $C_1$ and $C_2$, shared memory $M$ and given bound $k$.
This shared memory is the combination of memories in $C_1$ and $C_2$. When configurations $C_1$ and $C_2$ are final, equivalence requires $M_{C_1}=M_{C_2}=M$.
Being a symbolic execution tool, our prototype implementation is \textit{sound} (reports only true positives and true negatives) and \textit{bounded-complete} since it exhaustively and precisely explores all possible paths up to the given bound, which defines the number of consecutive function calls allowed.

Because of the infinite nature of proving equivalence\,---\,and even of disproving equivalence\,---\,of pure higher-order programs, a bounded exploration often does not suffice for automatic verification. For this reason, we implement simple enhancements that attempt to prune the state-space and/or prove that cycles have been reached to finitise the exploration for several examples in our testsuite. 
We currently have not implemented more involved up-to bisimulation enhancements, perhaps guided by user annotations, which we leave for future work. 
In particular we make use of: 
\begin{itemize}

\item \emph{Memoisation}, which caches configuration pairs. When bounded exploration reaches a memoised configuration pair, the tool does not explore any further outgoing transitions from this pair; these were explored already when the pair was added to the memoisation set.

\item \emph{Identity}, which deems two configurations in a pair equivalent when they are syntactically identical; no further exploration is needed in this case.

\item \emph{Normalisation}, which renames bound variables and symbolic constants before comparing configuration pairs for membership in the memoisation set. This also normalises the symbolic environments $\sigma$ in the configuration pairs.

\item \emph{Proponent call caching}, which caches proponent calls once the corresponding opponent return is reached. When the same call (same function applied to the same argument) is reached again on the same trace, it is immediately followed by the cached opponent return move. Performing this second call would not have materially changed the configuration, as the behaviour of the call is determined by the traces in the memory $M$ of the configuration.

\item \emph{Opponent call skipping}, which caches opponent calls once the corresponding proponent return is reached. If the same call is possible from later configurations with the same opponent knowledge, the call is skipped as the opponent cannot increase its knowledge by repeating the same call.

\item \emph{Stack-based loop detection}, which searches the stack component $K$ of a configuration for nested identical proponent calls. When this happens, it means that the configuration is on an infinite trace of interactions between opponent and proponent which will keep applying the same function indefinitely. We deem these configurations diverging.
\end{itemize}

Running our tool on the examples in this paper on an Intel Core i7 1.90GHz machine with 32GB RAM running OCaml 4.10.0 and Z3 4.8.10 on Ubuntu 20.04 we obtain the following three-trial average results: \cref{ex:pure}, deemed equivalent, 8ms;
\cref{ex:u's}, inequivalent, 3ms; \cref{ex:v's}, inequivalent, 4ms; \ref{ex:fix-curried-ineq}, inequivalent, 3ms.
For the entire benchmark of thirty seven program pairs, we successfully verify six equivalences and nineteen inequivalences with twelve inconclusive results in 471ms total time.
The complete set of examples is available in our online repository.






  \section{Conclusion}
  \label{sec:concl}
  We have proposed a technique which combines operational and denotational approaches in order to provide a (quotient-free) characterisation of contextual equivalence in call-by-value PCF. 
This technique provides the first fully abstract normal form bisimulation for this language.
We have justified several of our choices in designing our LTS via examples, and we believe the LTS is succinct in  not carrying more information than needed for completeness. 
Our technique gives rise to a sound and complete  technique for checking of \lang program equivalence, which we implemented into a bounded bisimulation checking tool.

After testing our tool implementation, we have found it useful for deciding instances of the equivalence problem. This is particularly true for inequivalences: the tool was able to verify most of our examples, including some which were difficult to reason about even informally.
Further testing and optimisation of the implementation are needed in order to assess its practical relevance, particularly on larger examples. Currently, the main limitation for the tool is the difficulty in establishing equivalences, as these typically entail infinite bisimulations and are hard to capture in a bounded manner. To address this, we aim to develop  up-to techniques~\cite{PousS11} and
(possibly semi-automatic) abstraction methods in order to finitise the examined bisimulation space.



  \bibliographystyle{splncs04}

  \bibliography{references}


  \clearpage
  \onecolumn
  \appendix[Soundness and completeness]
  
  \section{Soundness and Completeness}\label{apx:SC}

\subsection{Soundness of game-LTS}
In this section we show that the game-LTS is sound for contextual equivalence. Our argument proceeds by first introducing two notions of composition on the LTS (so called \emph{internal} and \emph{external}) and then by using compositionality to relate the semantics of an expression-in-context to traces of the expression and the context respectively (cf.~\cite{Laird07}). 

We start off with a definition and a couple of technical lemma on the game-LTS.
If $\C$ is an $O$-configuration, we also define:
\[
  \CP^+(\C) = \{ t\Epropret{()}   \mid \C\xtoo{t}\pconf{\_}{\_}{\emptyK}{\_}{()}{\emptyK}[\_]
\}.
\]
where note above that $t$ must be complete.

    
    \begin{lemma}\label{ref:playsSTD}
      Let $t$ be a play. Then both $\pview{t}$ and $\oview{t}$ are legal traces satisfying visibility.
    \end{lemma}
    \begin{proof}
    Proved using standard game semantics techniques (e.g.~\cite{Mccusker}).
      \end{proof}

\begin{lemma}\label{lem:extLTS}
Let $\C$ be an initial configuration and suppose that $\C\xtoo{t} \C'$. Then:
\begin{enumerate}
\item $t$ is a legal trace and if $\C'$ has components $\A,\conc$ then the names in $t$ are precisely $\A\cup\dom{\conc}$;
  \item for any name permutation $\pi$, $\C\xtoo{\pi\cdot t} \pi\cdot \C'$;
\item if $\C'=\oconf{\_}{\_}{\_}{t}{\_}{\vec\be}$ then the $P$-names in $\oview{t}$ are $\vec\be$ (in the same order);
\item if $\C'=\pconf{\_}{\_}{\_}{t}{\_}{\_}[\vec\al]$ then the $O$-names in $\pview{t}$ are $\vec\al$ (in the same order); moreover, for each play $t'o$ we have $\C'\xrightarrow{o}\C''$ for some $\C''$;
\item if $\C\xtoo \C''$ with $\pview{t}=\pview{t'}$ and both $\C',\C''$ are P-configurations not preceded by P-configurations then the expressions of $\C',\C''$ are the same; 
\item $t$ is a play;
  \item if, for some play $\hat t$, $\PVs(t)=\PVs(\hat t)$, then $\C\xtoo{\hat t}\C''$ for some $\C''$.
\end{enumerate}
\end{lemma}
\begin{proof}
  For~1,
the fact that $t$ is a legal trace follows by construction of the game-LTS, i.e.\ by the stack discipline and the fact that all names played as arguments are fresh.
Moreover, $\A$ and $\conc$ are populated precisely by names introduced in the trace $t$.

Claim~2 follows by nominal sets reasoning.

For~3, we proceed by induction on $|t|$. The case for $t\leq 1$ is straightforward. Otherwise, suppose that the last move in $t$ is some $p=\Epropret{\_}[\_]$ and $|t|>1$. By the stack discipline (for the $V$ component), we have that the current $V$ component is some $\vec\be,V'$, and that $p$ is returning the O-application that pushed $\vec\be$. By IH, the names in the O-view at that O-application were $\vec\be$, which suffices for our claim.
Finally, suppose that the last move in $t$ is some $p=\Epropapp{\al}{\_}{\_}$ and suppose that $\al$ in this application is decorated as $\al^{\vec\be'}$. By IH, at the point of introduction of $\al$ in $t$, the O-view contained precisely $\vec\be'$, which concludes the claim.

For~4, we proceed by induction on $|t|$. The case for $t\leq 1$ is straightforward. 
Otherwise, suppose that the last move in $t$ is some $o=\Eopret{\_}[\_]$ and $|t|>1$. By the stack discipline (for the $K$ component), we have that the current $K$ component is some $(E,\vec\al),K'$, and that  $o$ is returning the P-application that pushed $(E,\vec\al)$. By IH, the names in the P-view at that P-application were $\vec\al$, which suffices for our claim.
Now  suppose that the last move in $t$ is some $o=\Eopapp{\be}{\_}{\_}$ and suppose that $\be$ is mapped to some value decorated as $v^{\vec\al'}$. By IH, at the point of introduction of $\be$ in $t$, the P-view contained precisely $\vec\al'$, which concludes the claim.
Finally, if $t'o$ is a play then $o$ satisfies the $\mathsf{next}_O$ requirements, and if $o$ is an application of some $P$-name $\be$ then the latter is in the last component of $\C'$. Hence $\C'$ can play $o$.

For~5, we do induction on $n=|\pview{t}|$. The cases for $n\leq 1$ are straightforward. Suppose that the last move in $t,t'$ is some $o=\Eopret{\_}[\_]$ and $|\pview{t'}|>1$. By the stack discipline (for the $K$ component), we have that the current $K$ component is some $(E,\vec\al),K'$, and that  $o$ is returning the P-application that pushed $(E,\vec\al)$. Using the IH, and applying Tau rules on the same expression on both sides, we have that the expression $E\hole[\al v]$ triggering that P-application is common in the two cases. Hence, the expressions in $\C',\C''$ are the same.
Finally, suppose that the last move in $t,t'$ is some $o=\Eopapp{\be}{\_}{\_}$ and suppose that $\be$ is mapped to values  $v,v'$ in $\C',\C''$ respectively. By IH (and also applying Tau rules), at the point of introduction of $\be$ there was a common expression in both traces and, therefore, $v=v'$. This concludes the claim.

For~6, by claims~3 and~4 we have that applications refer to names in the corresponding views. 
Moreover, $O$-innocence follows by definition of the $\nextmove{O}{\cdot}$ function. We need to show $P$-innocence. Suppose that $t_1p_1,t_2p_2\sqsubseteq t$ with $\pview{t_1}=\pi\cdot \pview{t_2}$ for some $\pi$. Let $\C\xtoo{t_1}\C_1$ and $\C\xtoo{t_2}\C_2$, with $\C_1,\C_2$ not preceded by $P$-configurations, be the corresponding LTS reductions. By~2, $\C\xtoo{\pi\cdot t_2}\pi\cdot \C_2$ and hence, by~5, $\C_1$ and $\pi\cdot \C_2$ contain the same expressions. Thus, the Tau rules that can be applied on $\C_1$ and $\C_2$ are uniquely defined, and reach to configurations $\C_1',\C_2'$ such that $\C_1'$ and $\pi\cdot \C_2'$ contain the same expressions and are ready to trigger the same P-move, apart possibly from the selection of fresh $P$-names. By strong support lemma, we obtain that $\pview{t_1p_1}\sim\pview{t_2p_2}$.

For~7, we show that $\C$ produces every prefix $t'\sqsubseteq\hat t$, by induction on $|t'|$. If $|t'|\leq1$ then by hypothesis we have that $\C$ produces $t'$. Suppose now $t'=t''p$ with $t''$ non-empty. By IH, $\C\xtoo{t''}\C''$, and suppose that $\C''$ is not preceded by a $P$-configuration. By hypothesis, there are a $\pi$ and a prefix $\tilde t$ of $\pi\cdot t$ such that $\pview{\tilde t}=\pview{t''}$, and say $\C\xtoo{\tilde t}\tilde \C$, with the latter not preceded by a $P$-configuration. Then, by~5, $\C''$ and $\tilde \C$ contain the same expression and, hence, will play the same move $p$, assuming an appropriate choice of $\pi$, as required.
Finally, suppose $t'=t''o$ with $t''$ non-empty. By IH, $\C\xtoo{t''}\C''$ and $\C''$ is an $O$-configuration.
By~4, we have that $\C''\xrightarrow{o}{\C'''}$, as required.
\myqed
\end{proof}

\subsubsection{Semantic Composition}

We start by defining a notion of composition that combines the traces produced by two configurations. These are supposed to correspond to an expression and its context, but for now we will only require that the configurations satisfy a set of compatibility conditions.

Let $\phi$ be a partial bijection from $\ONs$ to $\PNs$. Given a trace $t$
with its names included in $\dom{\phi}\cup\rng{\phi}$, we define its \emph{dual} with respect to $\phi$, written $\bar{t}^\phi$, inductively by
$\bar{\emptyseq}^\phi=\emptyseq$ and:
\begin{align*}
  \overline{\Epropret{D}[\vec\be]}^\phi&=\Eopret{D}[\phi^{-1}(\vec\be)],&
  \overline{\Eopret{D}[\vec\al]}^\phi&=\Epropret{D}[\phi(\vec\al)] \\ 
  \overline{\Epropapp{\al}{D}{\vec\be}}^\phi&=\Eopapp{\phi(\al)}{D}{\phi^{-1}(\vec\be)}&
  \overline{\Eopapp{\be}{D}{\vec\al}}^\phi&=\Epropapp{\phi^{-1}(\be)}{D}{\phi(\vec\al)}.\end{align*}
We say configurations $\C$ and $\C'$ of {opposite polarity} are \emph{$\phi$-compatible} ($\C \asymp_\phi \C'$) if:
\begin{itemize}
\item their names are disjoint: $\A\cap \A'=\dom{\conc}\cap\dom{\conc'}=\emptyset$;
\item their traces are complementary up to $\phi$, i.e.\ $t=\bar{t'}^\phi$;
\item their stacks are compatible, written $(K,V) \asymp_\phi (K',V')$, which means:
\begin{itemize}
\item $\eval = \eval'=V=V'=\varepsilon$; or
\item $\eval = (E,\vec\al),\eval_1$ and $V' = \vec\be,V_1'$, with $\phi(\vec\al)=\vec\be$, and $(\eval_1,V) \asymp_{\phi'} (\eval',V_1')$ with $\phi'\subseteq\phi$; or
\item $V = \vec\be,V_1$ and $\eval' = (E,\vec\al),\eval_1'$, with $\phi(\vec\al)=\vec\be$, and $(\eval,V_1) \asymp_{\phi'} (\eval_1',V')$ with $\phi'\subseteq\phi$;
\end{itemize}
\item $\phi:\A\cup \A'\to\dom{\conc}\cup \dom{\conc'}$ and:
  \begin{itemize}
  \item for each $\al^{\vec\be}\in \A_1$, if $\conc_2(\phi(\al))=v^{\vec\al}$ then $\phi(\vec\al)=\vec\be$;
  \item for each $\al^{\vec\be}\in \A_2$, if $\conc_1(\phi(\al))=v^{\vec\al}$ then $\phi(\vec\al)=\vec\be$;
  \item if $\C_1,\C_2$ have last components $\vec\al,\vec\be$ (or $\vec\be,\vec\al$) then $\phi(\vec\al)=\vec\be$.
  \end{itemize}
\end{itemize}
With these definitions, we proceed to defining different notions of composition.

Suppose $\C \asymp_\phi \C'$. The following rules define the semantic composition of two configurations (symmetric rules $\textsc{App}_R,\textsc{Ret}_R$ are omitted).
\begin{align*}
{\begin{prooftree}
\Hypo{ \C_1 \xrightarrow{\tau} \C_1' }
\Hypo{ \C_2' = \C_2 }
\Infer2[ $\textsc{Int}_L$ ]
{ \C_1 \semcomp_\phi \C_2 \to' \C_1' \semcomp_\phi \C_2' }
\end{prooftree}}
\quad
{\begin{prooftree}
\Hypo{ \C_2 \xrightarrow{\tau} \C_2' }
\Hypo{ \C_1' = \C_1 }
\Infer2[ $\textsc{Int}_R$ ]
{ \C_1 \semcomp_\phi \C_2 \to' \C_1' \semcomp_\phi \C_2' }
\end{prooftree}}
\end{align*}
\begin{align*}
{\begin{prooftree}
\Hypo{ \C_1 \xrightarrow{\Epropapp{\al}{D}{\vec\be}} \C_1' }
\Hypo{ \C_2 \xrightarrow{\Eopapp{\be}{D}{\vec\al}} \C_2' }
\Hypo{ \phi(\al)=\be }
\Hypo{ \vec\al\notin \C_1\land\vec\be\notin \C_2 }
\Infer4[ $\textsc{App}_L$ ]
{ \C_1 \semcomp_\phi \C_2 \to' \C_1' \semcomp_{\phi\uplus[\vec\al\mapsto\vec\be]} \C_2' }
\end{prooftree}}
\end{align*}
\begin{align*}
{\begin{prooftree}
\Hypo{ \C_1 \xrightarrow{\Epropret{D}{\vec\be}} \C_1' }
\Hypo{ \C_2 \xrightarrow{\Eopret{D}{\vec\al}} \C_2' }
\Hypo{ \vec\al\notin \C_1\land\vec\be\notin \C_2 }
\Infer3[ $\textsc{Ret}_L$ ]
{ \C_1 \semcomp_\phi \C_2 \to' \C_1' \semcomp_{\phi\uplus[\vec\al\mapsto\vec\be]} \C_2' }
\end{prooftree}}
\end{align*}
We can show the following.
\begin{lemma}
If $\C_1\asymp_\phi \C_2$ and $\C_1\semcomp_\phi \C_2\to' \C_1'\semcomp_{\phi'} \C_2'$ then $\C_1'\asymp_{\phi'} \C_2'$.
\end{lemma}

\subsubsection{Composite Semantics and Internal Composition}

We now introduce the notion of composing LTS configurations \emph{internally}, which occurs when merging two compatible configurations into a single expression. The composition will ensure that the external functions (i.e.\ $O$-names) of each configuration are closed by the other one.
Hence, we will use extended configurations of the form:
\[
(\A_L,\A_R,\conc_L,\conc_R,\phi,e)\text{ or, for brevity, }(\vec \A,\vec \conc,\phi,e)
\]
such that $\A_L\cap \A_R=\dom{\conc_L}\cap\dom{\conc_R}=\emptyset$ and $\phi:\A_L\cup \A_R\to\dom{\conc_L}\cup\dom{\conc_R}$.
Our aim is to keep extended configurations in sync with composite configurations (i.e.\ those using $\semcomp$), and for this reason it will be useful to extend our expression syntax with an ``external return'' construct:
\[
e \ ::= \ \dots \mid \ret_L(e)\mid \ret_R(e)
  \]
We next define the semantics for extended configurations ($\compsemto$).
  \[\begin{array}{r@{\;\,}l@{\;\,}l}
      (\vec \A,\vec\conc,\phi,e)    & \compsemto & (\vec \A,\vec\conc,\phi,e') \;\ \text{if } e \to e'  \\
      (\vec \A,\vec\conc,\phi,E[\al v]) &\compsemto &
\begin{cases}
          (\vec \A\uplus_R\vec\al,\vec\conc\uplus_L[\vec\be\mapsto \vec v],\phi\uplus[\vec\al\mapsto\vec\be],E[\ret_L(e)]) & \text{if } \al\in \A_L, \app {\conc_L(\phi(\al))} {D\hole[\vec \al]} \funred e \\
          (\vec \A\uplus_L\vec\al,\vec\conc\uplus_R[\vec\be\mapsto \vec v],\phi\uplus[\vec\al\mapsto\vec\be],E[\ret_R(e)]) & \text{if } \al\in \A_R, \app {\conc_R(\phi(\al))} {D\hole[\vec \al]} \funred e
          \end{cases}\\
        (\vec \A,\vec\conc,\phi,E[\ret_L( v)]) & \compsemto & (\vec \A\uplus_L\vec\al,\vec\conc\uplus_R[\vec\be\mapsto\vec v],\phi\uplus[\vec\al\mapsto\vec \be],E[D\hole[\vec \al]]) \\
        (\vec \A,\vec\conc,\phi,E[\ret_R( v)]) & \compsemto & (\vec \A\uplus_R\vec\al,\vec\conc\uplus_L[\vec\be\mapsto\vec v],\phi\uplus[\vec\al\mapsto\vec \be],E[D\hole[\vec \al]]) 
  \end{array}\]
where, in the last four rules, $(D,\vec v) \in \ulpatt(v)$.

We continue by defining the \boldemph{internal composition} of compatible configurations $\C_L \asymp_\phi \C_R$. We define the internal composition $\C_L \confcomp_\phi \C_R$ to be a configuration in our new composite semantics by pattern matching on the configuration polarity and evaluation stacks according to the following rules. 

\paragraph*{Initial Configuration:}
\begin{align*}
\C_L &= \pconf{\emptyset}{\cdot}{\emptyK}{\emptytrace}{e}{\emptyK}[\emptyseq]\\
\C_R &= \oconf{\emptyset}{\cdot}{(E\hole,\emptyseq)}{\emptytrace}{\emptyK}{\emptyseq}\\
\C_L \confcomp \C_R &= (\emptyset,\emptyset,\cdot,\cdot,\cdot,E[e])
\end{align*}

\paragraph*{Interim Configuration (case PO):}
\begin{align*}
\C_L &= \pconf{\A_L}{\conc_L}{K_L}{t}{e}{\vec\be_L,V_L}[\vec\al_L]\\
\C_R &= \oconf{\A_R}{\conc_R}{K_R}{t}{V_R}{\vec\be_R}\\
\C_L \confcomp_\phi \C_R &= (\A_L,\A_R,\conc_L,\conc_R,\phi,(K_L\confcomp_R K_R)[e])
\end{align*}

\paragraph*{Interim Configuration (case OP):}
\begin{align*}
\C_L &= \oconf{\A_L}{\conc_L}{K_L}{t}{V_L}{\vec\be_L}\\
\C_R &= \pconf{\A_R}{\conc_R}{K_R}{t}{e}{\vec\be_R,V_R}[\vec\al_R]\\
\C_L \confcomp_\phi \C_R &= (\A_L,\A_R,\conc_L,\conc_R,\phi,(K_L\confcomp_L K_R)[e])
\end{align*}
where $K_L \confcomp_{L/R} K_R$ is a single evaluation context resulting from the composition of compatible stacks, which we define as follows:
\begin{align*}
\emptyK \confcomp_{L/R} \emptyK &= \hole\\
((E\hole,\vec\al),K_L) \confcomp_L K_R &= (K_L \confcomp_R K_R)[E[\ret_L(\hole)]]\\
K_L \confcomp_R ((E\hole,\vec\al),K_R) &= (K_L \confcomp_L K_R)[E[\ret_R(\hole)]]
\end{align*}
Notice that there is only one case for initial configurations, and that is because the game must start from an opponent-proponent configuration where stacks are empty.

\subsubsection{Bisimilarity of Semantic and Internal Composition}
 We begin by defining (weak) bisimilarity for the semantic and internal composition. A set $\mathcal{R}$ with elements of the form $(X_1,X_2)$, {where $X_1$ is a pair of the form $\C_L\semcomp_\phi \C_R$ and $X_2$ is from the composite semantics}, is an \emph{si-bisimulation} if for all $(X_1,X_2) \in R$:
\begin{itemize}
\item if $X_1 \to' X_1'$ then $X_2 \compsemto^* X_2'$ and $(X_1',X_2') \in \mathcal{R}$;
\item if $X_2 \compsemto X_2'$ then $X_1 \to'^* X_1'$ and $(X_1',X_2') \in \mathcal{R}$.
\end{itemize}
We say that $X_1,X_2$ are \emph{si-bisimilar}, and write $X_1 \simsi X_2$, if there is an si-bisimulation $\mathcal{R}$ such that $X_1 \mathcal{R} X_2$.


{\lemma{\label{bisim}}{Given configurations $\C_L \asymp_\phi \C_R$, it is the case that $(\C_L \semcomp_\phi \C_R) \simsi (\C_L \confcomp_\phi \C_R)$.}}

\subsubsection{Soundness}

To prove soundness of the game-LTS semantics, we want to show that syntactic composition can be obtained from semantic counterpart and vice versa. We have si-bisimilarity between semantic and internal composition, we only need to show that internal composition is related to syntactic composition under some notion of equivalence. 

\begin{lemma}\label{lem:composition}
  Let $\vdash e:T$ by any expression and  $E$ any evaluation context such that $\vdash E[e]:\mathsf{unit}$,
  and let
    \[
\C_L = \pconf{\emptyset}{\cdot}{\emptyK}{\emptytrace}{e}{\emptyK}[\emptyseq]
\ \text{ and } \
\C_R = \oconf{\emptyset}{\cdot}{(E\hole_T,\emptyseq)}{\emptytrace}{\emptyK}{\emptyseq}.
\]
Then,
  $E[e]\Downarrow$ iff there is a complete play $t$ and some $\phi$ such that
    $t\in\CP(\C_L)$ and $\bar{t}^\phi\Epropret{()}\in\CP^+(\C_R)$.
\end{lemma}
\begin{proof}
  Suppose $E[e]\Downarrow$, and 
recall that $\C_L\confcomp \C_R=(\emptyset,\emptyset,\cdot,\cdot,\cdot,E[\ret_R(e)])$.
  Then: 
  \begin{enumerate}
  \item By inspection of the composite semantics, we have that $\C_L\confcomp \C_R$ reaches $()$.
	\item By si-bisimilarity (Lemma~\ref{bisim}) we have that $\C_L \semcomp \C_R$ reaches $()$, say with final name bijection $\pi$. 
	\item By definition of semantic composition, we know there are traces $t \in \CP(\C_L)$ and $t'\Epropret{()}\in\CP^+(\C_R)$  such that $t'=\bar{t}^\phi$.
\end{enumerate}
Conversely, suppose there is complete $t\in\CP(\C_L)$ such that $\bar{t}^\phi\Epropret{()}\in\CP^+(\C_R)$. Then:
\begin{enumerate}
	\item By definition of semantic composition we have that $\C_L \semcomp \C_R$ reaches $()$, with final bijection $\phi$.
	\item By si-bisimilarity (Lemma~\ref{bisim}) we have that $\C_L\confcomp \C_R$ reaches $()$.
	\item By inspection of the composite semantics, we obtain that $E[e]$ reaches $()$. 
\end{enumerate}\myqed
\end{proof}

\begin{corollary}[Complete-play Soundness]
Given $\vdash e_1,e_2:T$, if  $\CP(\C_{e_1})=\CP(\C_{e_2})$ then $e_1\cxteq e_2$.
\end{corollary}
\begin{proof}
  Suppose there is some evaluation context such that $E[e_1]\Downarrow$ and $E[e_2]\Uparrow$. Then, by previous Lemma, there is some $t\in\CP(\C_{e_1})$ and $\bar{t}^\phi\in\CP(\C_{E})$. By the same Lemma, and the fact that $E[e_2]\Uparrow$, we have that $t\notin\CP(\C_{e_2})$.\myqed
\end{proof}

\begin{proposition}[Soundness]
Given $\vdash e_1,e_2:T$, if  $\OVTLs(\C_{e_1})=\OVTLs(\C_{e_2})$ then $e_1\cxteq e_2$.
\end{proposition}
\begin{proof}
  Suppose there is some applicative context $E$ such that $E[e_1]\Downarrow$ and $E[e_2]\Uparrow$. Then, by Lemma~\ref{lem:composition}, there is some $t\in\CP(\C_{e_1})\setminus\CP(\C_{e_2})$ and $\phi$ such that $\bar{t}^\phi\Epropret{()}\in\CP^+(\C_{E})$. Since $E$ is applicative, $t$ must be top-linear.
  Taking $S=\{\pi\cdot\oview{t'}\mid t'\sqsubseteq t\}$, we have $S\in\OVTLs(\C_{e_1})$.
Let us suppose that $S\in\OVTLs(\C_{e_2})$. Then, there is $\hat t\in\CP(\C_{e_2})$ with
  $\{\pi\cdot\oview{t'}\mid t'\sqsubseteq \hat{t}\}=S$. Then, by Lemma~\ref{lem:extLTS}(7) we have that $\bar{\hat{t}}^{\phi}\in\Tr(\C_{E})$. Now observe that, since $t,\hat t$ are complete:
  \[
\bar{t}^{\phi} = o_0p_1\dots o_1\cdots p_n\dots o_n\text{ and }
\bar{\hat{t}}^{\phi} = \hat o_0\hat p_1\dots\hat o_1\cdots \hat p_{\hat n}\dots\hat o_{\hat n}
\]
where, for each $i>0$, $o_i$ returns $p_i$ and $\hat o_i$ returns $\hat p_i$. By hypothesis, we have that $o_0=\hat o_0$ and thus P-innocence implies that $p_1\sim\hat p_1$. WLOG we can assume that $p_1=\hat p_1$(since otherwise we can consider permutation variants of $t,\hat t$ which agree on the corresponding fresh names). Again by hypothesis there is a prefix $t'$ of $\bar{t}^\phi$ such that $\pview{t'}=\hat o_0\hat p_1\hat o_1=o_0p_1\hat o_1$. By O-innocence (on $S$) we have
$o_1\sim \hat o_1$ and, WLOG,
$o_1=\hat o_1$. Applying the same reasoning repeatedly we obtain that  $\pview{\bar{t}^{\phi}} =\pview{\bar{\hat{t}}^{\phi}}$, and therefore $\bar{\hat t}^\phi\Epropret{()}\in\CP^+(\C_E)$.
  Hence, by Lemma~\ref{lem:composition} we have that $E[e_2]\Downarrow$, a contradiction. Thus, $S\notin\OVTLs(\C_{e_2})$. 
\myqed
\end{proof}

\subsection{Definability and Completeness for Game-LTS}

We prove completeness via a definability argument which follows the lines of game-semantics definability proofs (e.g.~\cite{HO,Tze09}).
We show that, for each complete play $t$, the set $\OVs(t)$ has a matching evaluation context $E$ such that $\C_E$ realises (the \emph{duals} of) all plays in $\OVs(t)$.

For the definability argument it will be useful to consider open terms, where the open variables are instantiated with (possibly abstract) values. We therefore
redefine initial configurations to be of  the form $\langle\Delta\vdash e:T\rangle$ and  
extend the game-LTS of \cref{fig:lts3} with initialisation moves ($\Eopapp{?}{D}{\vec\al}$) and the initialisation rule:
\[\tag{\textsc{Init}}
  \langle\Delta\vdash e:T\rangle\xrightarrow{\Eopapp{?}{D}{\vec\al}}
  \pconf{\vec\al}{\cdot}{\emptyK}{\emptytrace}{e[\vec\al/\vec x]}{\emptyseq}[\vec\al]\qquad\text{if }(D,\vec\al)\in\ulpatt(T_1*\cdots*T_n)
  \]
  assuming $\Delta=\{x_1:T_1.\dots,x_n:T_n\}$. Accordingly, \emph{extended complete traces} are given by the grammar:
  \[
ECT  \ \to \ \Eopapp{?}{D}{\vec\al}\,CT_P
\]
whereas \emph{extended traces} are prefixes thereof. 
The notions of $O$-view and $P$-view are defined as in plain traces, with the caveat that the second move in an extended trace is justified by the first move (if they exist). Hence, \emph{extended plays} are legal extended traces satisfying the conditions of Definition~\ref{def:plays}.

We denote configurations in the extended LTS by $\K$ and variants for clarity. We write $\CEP(\Delta\vdash e:T)$ for the set of complete extended plays produced by the extended LTS starting from 
$\langle\Delta\vdash e:T\rangle$. We then let
\[
  \EPVs(\Delta\vdash e:T)=\{\pview{t'}\mid t'\sqsubseteq t\in\CEP(\Delta\vdash e:T)\land |t'|\,\text{even}\}.
\]
Note that we require that only even-length $P$-views are in $\EPVs$ (this is for technical convenience).

\begin{definition}
  We call a set of extended plays $\F$ a \emph{viewfunction} if:
  \begin{enumerate}
  \item for all $t\in \F$, $\pview{t}=t$, and $\F$ is even-prefix closed and also closed under name permutations (i.e.\ from $\Perm$);
  \item if $tm_1,tm_2\in \F$ then $tm_1\sim tm_2$;
    \item if $m_1t_1,m_2t_2\in \F$ then $m_1\sim m_2$.
    \end{enumerate}
Given types $\vec T,T$,
    we write $\vec T\vdash\F:T$ if either $\F$ is empty or it contains $\Eopapp{?}{D_0}{\vec\al_0}\Epropret{D_1}[\vec\be_1]$ such that $(D_0,\vec\al_0)\in\ulpatt{T_1*\cdots*T_n}$ and $D_1[\vec\be]:T$.

Finally, we call $\F$ \emph{finite-orbit} if its set of orbits (under permutations):
  \[
\mathcal{O}(\F) = \{ [t] \mid t\in \F\}
\]
is finite. We then set $\|\F\| = |\mathcal{O}(\F)|$.
\end{definition}

Thus, a viewfunction is a set of extended $P$-views representing a deterministic proponent behaviour. Condition~3 in the definition above imposes such a notion of determinacy (\emph{P-innocence}): each next opponent move is determined by the preceding $O$-view.
Condition~4 imposes uniqueness of the initialising move and, combined with $P$-innocence, ensures that typings are unique for non-empty $\F$. 

\begin{proposition}[Definability]\label{prop:definability}
For any finite-orbit $\vec T\vdash \F:T$ there is a term $\Delta\vdash e:T$ evaluation context $E$ such that $\EPVs(\Delta\vdash e:T)={\F}$. 
\end{proposition}
\begin{proof}
  Suppose that $\Delta=\{x_1:T_1,\dots,x_n:T_n\}$.
  We do induction on $\|\F\|$. For the base case ($\F=\emptyset$) we can set $E\defeq \bot_T$. Suppose now $\|\F\|>0$, so there is some $\Eopapp{?}{D_0}{\vec\al_0}\,m_0\in{\F}$.
  If there is some evaluation term $\Delta\vdash e':T$ such that
  \[\label{IH}\tag{$*$}
  {\F} = \{ mt\in\EPVs(\Delta\vdash e':T)\mid m\sim\Eopapp{?}{D_0}{\vec\al_0}\}
  \]
  then we can set $e\defeq{\cond{\vec x=D_0}{e'}{\bot_T}}$, where $\vec x=D_0$ is a (definable) macro that checks whether the ultimate pattern of $\vec x$ is $D_0$.

  We next produce $e'$. The move $m_0$ can either be a return or an application of some $\al_{0i}$. Suppose first that $m_0=\Epropret{D_1}[\vec\be]$, and let $\vec\be=\be_1 \cdots \be_k$. For each $j$, assuming $\be_j$ has function type $T_j=T_j'\to T_j''$, consider the set:

  \[
\F_j=\bigcup\{ [\,\Eopapp{?}{(D_0,D)}{\vec\al_0,\vec\al}\,t\,]\mid \Eopapp{?}{D_0}{\vec\al_0}\Epropret{D_1}[\vec\be]\Eopapp{\be_j}{D}{\vec\al}\,t\in\F\}
\]
We can see that $\F_j$ is a viewfunction and, moreover, $\|\F_j\|<\|\F\|$, so by IH we obtain a term $\Delta,y:T_j'\vdash e_j:T_j''$ such that $\EPVs(e_j)=\F_j$.
Thus, taking 
$
e'\defeq {D_1[\overrightarrow{\lambda y.e_j}]}
$
we can satisfy~\eqref{IH}.

Finally, suppose that $m_0=\Epropapp{\al_{0i}}{D_1}{\vec\be}$, with $\vec\be=\be_1\cdots\be_k$, and assume that $\be_j:T_j'\to T_j''$ for each $j=1,\dots,k$, while $\al_{0i}:T_0'\to T_0''$.
In this case, proponent responds by calling one of the opponent functions ($\al_{0i}$). We consider the following sets of P-views which determine proponent's behaviour after each subsequent opponent move:
\begin{align*}
{\F_R} &= \bigcup \{ [\,\Eopapp{?}{(D_0,D)}{\vec\be_0,\vec\be}\,t  \,] \mid \Eopapp{?}{D_0}{\vec\al_0}\Epropapp{\al_{0i}}{D_1}{\vec\be}\Eopret{D}[\vec\al]\,t\in\F\} 
  \\
{\F_j} &=\bigcup\{ [\,\Eopapp{?}{(D_0,D)}{\vec\be_0,\vec\be}\,t\,]\mid \Eopapp{?}{D_0}{\vec\be_0}\Epropapp{\al_{0i}}{D_1}{\vec\be}\Eopapp{\be_j}{D}{\vec\al}\,t\in{\F}\}
\end{align*}
we can see that each of the above has smaller measure than $\F$ hence, by applying the IH, we obtain terms $\Delta,y:T_0''\vdash e_R:T$ and $\Delta,z:T_j'\vdash e_j:T_j''$ (for each $j$). Then, taking
\[
  e'\defeq
    \elet{y}{\piD{D_0}{i}(x)(D_1[\overrightarrow{\lambda z.e_j}])}{e_R}
\]
we can satisfy~\eqref{IH}.
Here, $\piD{D_0}{i}(x)$ is a macro obtaining the $i$-th function component from an $x$ with ultimate pattern $D_0$.
This completes the proof.
\end{proof}

The remainder of this section follows closely~\cite{ChurchillEtal2020}.

\begin{definition}
Suppose $\Sigma,\Sigma'$ are sets of $P$-starting $O$-views. We let $\Sigma\leq\Sigma'$ if for all $\F\in\Sigma$ there is $\F'\in\Sigma'$ such that $\F'\subseteq\F$. 
  \end{definition}

Let us fix a bijection $\phi:\ONs\overset{\cong}{\to}\PNs$. For each set of plays $S$, we shall write $\dual[]{S}$ for the set $\{\dual{t}\mid t\in S\}$. 

\begin{lemma}
Given $\vdash e_1,e_2:T$, if   $e_1\cxteq e_2$ then $\OVs(\C_{e_1})\leq\OVs(\C_{e_2})\leq\OVs(\C_{e_1})$.
\end{lemma}
\begin{proof}
  By symmetry, it suffices to show one inclusion. Suppose $e_1\cxteq e_2$ and let
  $t\in\CP(\C_{e_1})$ be a top-linear complete play. Consider the set of extended plays:
  \[
\F = \{\Eopapp{?}{D_0}{\phi(\vec\be)}\,\dual{t'}\mid \Epropret{D_0}[\vec\be]\,t'\in\OVs(t) \land \, |t'|\!\text{ odd}\}\cup\{ \Eopapp{?}{D_0}{\phi(\vec\be)}\,\dual{t'}\Epropret{()}\mid \pview{t}=\Epropret{D_0}[\vec\be]\,t'\}
\]
By determinacy of the LTS (for both $O$ and $P$), we have that $\F$ is a viewfunction. Hence, by Proposition~\ref{prop:definability}, there is a term $x:T\vdash e:\Unit$ such that $\F=\EPVs(e)$.
Now, define the context:
\[
E\defeq \elet{x}{\hole_T}{e}
\]
By inspection of the LTS rules, we can see that the plays in $\CP(\C_E)$ and the extended plays in $\CEP(e)$ are the same, modulo adjusting the first move of extended plays (from an initialisation move to an $O$-return one) and removing the last one ($\Epropret{()}$).
Hence, $\PVs(\C_E)=\dual[]{\OVs(t)}$.
Moreover, $\C_E$ produces all traces in $\PVs(\dual t)$, and therefore $\dual t\in\CP(\C_E)$. By simulating the reduction producing $\dual t$ in the extended LTS (for $e$), we can see that $\dual t\Epropret{()}\in\CP^+(\C_E)$. Hence, by Lemma~\ref{lem:composition}, $E[e_1]\Downarrow$ and thus, by hypothesis, $E[e_2]\Downarrow$.
Again, by Lemma~\ref{lem:composition} and using the fact that the LTS is closed with respect to name permutations, there is a play $s\in\CP(\C_{e_2})$ such that $\dual{s}\Epropret{()}\in\CP^+(\C_E)$.
By $\dual s\in\CP(\C_E)$, we have that $\PVs(\dual s)\subseteq\PVs(\C_E)=\dual[]{\OVs(t)}$ and, hence, $\OVs(s)\subseteq\OVs(t)$. Thus, $\OVs(\C_{e_1})\leq\OVs(\C_{e_2})$.
\end{proof}

\begin{lemma}
It $\OVs(\C_{e_1})\leq\OVs(\C_{e_2})\leq\OVs(\C_{e_1})$ then $\OVs(\C_{e_1})=\OVs(\C_{e_2})$.
\end{lemma}
\begin{proof}
  Let $S_1\in\OVs(\C_{e_1})$. By hypothesis, there are $T\in\OVs(\C_{e_2})$ and $S_2\in\OVs(\C_{e_1})$ such that $S_2\subseteq T\subseteq S_1$. Let $t_1,t_2\in\CP(\C_{e_1})$ be such that $S_i=\OVs(t_i)$.
  Suppose that $S_1\neq S_2$, so $t_1\neq t_2$. As the LTS is $P$-deterministic, there is an odd-length play $t$ and $o_1\not\sim o_2$ such that $to_1,to_2\in\Tr(\C_{e_1})$ and $\oview{t}o_i\in\OVs(t_i)$.
  In particular, $\oview{t}o_2\in S_2\subseteq S_1$, so $\oview{t}o_2\in S_1$. But then $\oview{t}o_1,\oview{t}o_2\in S_1$, contradicting $O$-determinacy. Thus, $S_1=S_2=T$ and, hence,
  $S_1\in\OVs(\C_{e_2})$. 
\end{proof}

\begin{proposition}[Completeness]
Given $\vdash e_1,e_2:T$, if   $e_1\cxteq e_2$ then $\OVTLs(\C_{e_1})=\OVTLs(\C_{e_2})$.
\end{proposition}
\begin{proof}
  It suffices to show one inclusion. Suppose
 $e_1\cxteq e_2$ so
by the previous two lemmata,
$\OVs(\C_{e_1})=\OVs(\C_{e_2})$. Let $t\in\CP(\C_{e_1})$ be top linear, so $\OVs(t)\in\OVs(\C_{e_2})$, say $\OVs(t)=\OVs(s)$ for some $s\in\CP(\C_{e_2})$. As all the plays in $\OVs(s)$ are top-linear, so is $s$.
Hence, $\OVs(t)\in\OVTLs(\C_{e_2})$.
  \end{proof}

\subsection{Correspondence of LTS and game-LTS}

\begin{figure*}[t] 

  \[\begin{array}{@{}cllll@{}}
    \irule[PropRetBarb][propretbarb]{
      (D,\vec v) \in \ulpatt(v)
      \\
      t' = t+\Epropret{D}[\vec\be]
      \\
      \conc'=\conc\uplus[\vec\be\mapsto\vec v^{\vec\al}]
    }{
      \pconf{A}{\conc}{\emptyK}{t}{v}{\vec\be'}[\vec\al]
      \trans{\Epropret{D}[\vec \be]}   
      \oconf{\emptyset}{\conc'}{\emptyK}{t'}{\cdot}{\vec \be}
    }
    \\[2em]
    \irule[OpAppBarb][opappbarb]{
       \be_i : T_1 \arrow T_2
      \\
      \conc(\be_i)=v^{\vec\al'}
      \\
        (D,\vec \al) \in \ulpatt(T_1)
       \\
      t'=t+\Eopapp{\be_i}{D}{\vec \al}
    }{
      \oconf{A}{\conc}{\emptyK}{t}{\emptyK}{\vec\be}
      \trans{\Eopapp{\be_i}{D}{\vec\al}}
      \pconf{A\uplus\vec\al^{{\varepsilon}}}{\conc}{\emptyK}{t'}{e}{\emptyK}[\vec\al']
    }
    \\[2em]
    \irule[PropRet][propret]{
      (D,\vec v) \in \ulpatt(v)
      \\
      K\not=\emptyK
      \\
      t' = t+\Epropret{D}[\vec\be]
      \\
      \conc'=\conc\uplus[\vec\be\mapsto\vec v^{\vec\al}]
    }{
      \pconf{A}{\conc}{K}{t}{v}{\vec\be',V}[\vec\al]
      \trans{\Epropret{D}[\vec\be]}
      \oconf{A}{\conc'}{K}{t'}{V}{\vec\be',\vec\be}
    }
    \\[2em]
    \irule[OpApp][opappf]{
       [\Eopapp{\be_i}{D}{\vec \al}]\subseteq \nextmove{O}{t}
      \\
       \be_i : T_1 \arrow T_2
      \\
      \conc(\be_i)=v^{\vec\al'}
      \\
        (D,\vec \al) \in \ulpatt(T_1)
       \\
      t'=t+\Eopapp{\be_i}{D}{\vec \al}
      \\
      K\neq\emptyK
    }{
      \oconf{A}{\conc}{K}{t}{V}{\vec\be}
      \trans{\Eopapp{\be_i}{D}{\vec\al}}
      \pconf{A\uplus\vec\al^{{\vec \be}}}{\conc}{K}{t'}{e}{\vec \be,V}[\vec\al',\vec\al]
    }
  \end{array}\]
  \hrule
  \caption{The Game Labelled Transition System, top-linear version (\textsc{PropTau,PropApp,OpRet} rules as in Figure~\ref{fig:lts3}).}\label{fig:lts3.1}
\end{figure*}

We next show that our (plain) LTS and game-LTS produce the same notion of term equivalence.
In this section, by game-LTS we refer to the LTS of Figure~\ref{fig:lts3.1}, so in particular all plays examined will be top-linear (cf.\ Lemma~\ref{lem:lts3.1}).
Moreover, by plain LTS we intend the LTS of Figure~\ref{fig:lts} albeit with the modification that in rules \textsc{PropRet, PropApp, OpRet, OpApp} the transition is labelled with the corresponding move (i.e.\ the one that ends up in the trace stored in the target configuration) and not with $\tau$. In particular, only transitions triggered by \textsc{PropTau} are labelled with $\tau$.
We can see that this modification does not essentially alter the LTS, and in particular it produces the same $M$-components.

\begin{lemma}\label{lem:lts3.1}
  Let $\CP'(\C)$ be the complete plays produced from $\C$ using the rules of the LTS in Figure~\ref{fig:lts3.1}. Then,
$\OVTLs(\C) = \{ \OVs(t) \mid t\in\CP'(\C)\}$.
\end{lemma}
We will show that there is a translation from one LTS to the other that preserves traces and the functions $M$.

Given components $\A,\conc,t$ from a reachable configuration (from the game-LTS), we shall define corresponding plain-LTS components and a function $\psi$ from $\ONs$ to abstract function names:\footnote{Note that, in this section, all traces are $P$-starting.}
\[
   ( A,M,\hat t,\hat s,\psi)\in (\A,\conc,t)^\circ.
 \]
 In particular, for each $\al\in\dom{\psi}$, $\psi(\al)=\alpha^i$, for some $\alpha,i$. We may write
 \[
\psi(\al) \funto \alpha\text{ if $\psi(\al)=\alpha^i$ for some $i$.}
   \]
We call a triple $(\A,\conc,t)$ \emph{compatible} if
\[
  \ONs(\conc)\subseteq\dom{\A}=\ONs(t)\ \land\
  \dom{\conc}=\PNs(t)
\]
where $\ONs(X)$ are the $O$-names featuring in $X$, and similarly for $\PNs(X)$.

First,
given any incomplete top-linear play $t$ with $|t|>1$, we can split $t$ as:
\[
t = t_{{cp}}\,o\,t_{{lo}}
\]
where $t_{cp}$ a complete trace, $o$ is a top-level $O$-call and $t_{lo}$ contains no top-level moves.
Then, we
let $\oviewpsi[\circ]{t}$ be the suffix $t'$ of $\oview{t}$ satisfying the condition:
\[
  \oview{t} \ = \ \oview{t_{cp}\,o}\cdots t'
\]
such that $t'$ contains exactly one P-call, which is at its start. If $t$ is complete or $t_{lo}$ is empty, then $\oviewpsi[\circ]{t}=\emptytrace$.

\begin{lemma}\label{lem:extLTS2}
Let $\C$ be an initial configuration and suppose that $\C\xrightarrow{t}\!\!\!\!\!\to \C'$ in the LTS of Figure~\ref{fig:lts3.1}. Then:
\begin{enumerate}
\item $t$ is a play and if $\C'$ has components $\A,\conc$ then the names in $t$ are precisely $\A\cup\dom{\conc}$;
  \item for any name permutation $\pi$, $\C\xrightarrow{\pi\cdot t}\!\!\!\!\!\to \pi\cdot \C'$;
\item if $\C'=\oconf{\_}{\_}{\_}{t}{\_}{\vec\be}$ then the $P$-names in $\oviewpsi[\circ]{t}$ are $\vec\be$ (in the same order).
\end{enumerate}
\end{lemma}

We shall also be using the following lemma regarding name permutations (cf.~\cite{Tze09}).

\begin{lemma}\label{lem:SSL}
Given sequences of moves $t_1\sim t_2$ and names $x_1,x_2$ of the same kind (either both $O$-names or both $P$-names), if $x_i$ is fresh for $t_i$ ($i=1,2$) then $t_1x_1\sim t_2x_2$.
\end{lemma}

\begin{definition} Given compatible $(\A,\conc,t)$, we
define $(\A,\conc,t)^\circ$ by induction on $|t|$. For the base case:
\[
  (\A,\conc,\emptytrace)^\circ = \{(\emptyset,\cdot,\emptytrace,\emptytrace,\cdot)\}
\]
If $t=t'+m$ then, for each $( A',M',\hat t',\hat s',\psi')\in(\A',\conc',t')^\circ$ where $\A',\conc'$ the restrictions of $\A,\conc$ respectively to the names in $t'$,
we include in $(\A,\conc,t)^\circ$ a triple $( A,M,\hat t,\hat s,\psi)$ defined 
by case analysis on $m$:
\begin{itemize}
\item If $m=\Eopret{D}[\vec\al]$ then 
  \begin{itemize}
  \item $M= M'[\hat t'+\lopret{D[\vec\alpha]}]$ where: if $M$ is defined in $\hat t'$ then we require that $\nextmove{M}{\hat t'}=\lopret{D[\vec\alpha]}$ (for some $\vec\alpha$), otherwise $\vec\alpha$ are fresh;
  \item $ A= A'\uplus\vec\alpha^{j,\vec v}$, for the least $j$ such that $\vec\alpha^{j,\ldots}\notin  A'$,
    and $\vec v=\psi'(\conc(\vec\be))$ with $\vec\al^{\vec\be}\in \A$;
  \item $\hat s=\hat s'$ and $\psi=\psi'[\vec\al\mapsto\vec\alpha^j]$;
    \item if $t''\sqsubseteq t$ ends in the last open call in $t$ (which must be an $O$-move) and $(\A'',\conc'',t'')^\circ=( A'',M'',\hat t'',\psi'')$ then $\hat t=\hat t''$.
  \end{itemize}
\item If $m=\Eopapp{\be}{D}{\vec\al}$ with $\be$ the $i$-th $P$-name in $\oviewpsi[\circ]{t}$, then
  \begin{itemize}
      \item if $t'$ is complete then $M=M'$, otherwise
  \ $ M= M'[\hat t'+\lopapp{i}{D[\vec\alpha]}]$, where: if $M$ is defined in $\hat t'$ then we require that $\nextmove{M}{\hat t'}=\lopapp{i}{D[\vec\alpha]}$ (for some $\vec\alpha$), otherwise $\vec\alpha$ are fresh;
  \item $ A= A'\uplus\vec\alpha^{j,\vec v}$, for the least $j$ such that $\vec\alpha^{j,\ldots}\notin  A'$,
    and $\vec v=\psi'(\conc(\vec\be))$ with $\vec\al^{\vec\be}\in \A$;
  \item if $t'$ is complete then $\hat s=\hat s'+\lopapp{i}{D[\vec\alpha]}$, otherwise $\hat s=\hat s'$; \item $\psi=\psi'[\vec\al\mapsto\vec\alpha^j]$;
    \item if $t'$ is complete then $\hat t=\emptytrace$, otherwise $\hat t=\hat t'+\lopapp{i}{D[\vec\alpha]}$.
  \end{itemize}
\item If $m=\Epropret{D}[\vec\be]$ then
  \begin{itemize}
  \item if $t$ is complete then $M= M'$, $ A= A'$, $\psi=\psi'$, $\hat s=\hat s'+\lpropret{D}$ and $\hat t=\emptytrace$;
  \item otherwise, $M= M'[\hat t'+\lpropret{D}]$, $ A= A'$, $\psi=\psi'$, $\hat s=\hat s'$ and $\hat t=\hat t'+\lpropret{D}$.
  \end{itemize}
\item If $m=\Epropapp{\al}{D}{\vec\be}$ then, assuming $\psi'(\al)\funto\alpha$,
  \begin{itemize}
  \item $M= M'[\lpropapp{\alpha}{D}]$, $ A= A'$, $\psi=\psi'$,  $\hat t=\lpropapp{\alpha}{D}$ and $\hat s=\hat s'$.
  \end{itemize}
\end{itemize}
Given a play $t$ we then define:
\[
  (t)^\circ = \{(M,\hat t,\hat s,\psi)\mid\exists \A,\conc, A.\   (\A,\conc,t)\text{ compatible}\,\land(\A,\conc,t)^\circ = ( A,M,\hat t,\hat s,\psi)\}.
\]
\end{definition}

Observe above that the components $M,\hat t,\hat s,\psi$ are defined using  solely $t$, and thus $(t)^\circ$ is well defined.
Moreover, any two elements of $(\A,\conc,t)$ are equal up to permutation of $O$-names.

\begin{lemma}
  Given compatible $(\A,\conc,t),(\A',\conc',t)$ and some $(M,\hat t,\hat s,\psi)$:
  \begin{enumerate}
  \item $\exists A.( A,M,\hat t,\hat s,\psi)\in(\A,\conc,t)^\circ\iff\exists A'.( A',M,\hat t,\hat s,\psi)\in(\A',\conc',t)^\circ$;
    \item $\forall( A_1,M_1,\hat t_1,\hat s_1,\psi_1),( A_2,M_2,\hat t_2,\hat s_2,\psi_2)\in(\A,\conc,t)^\circ.\,\exists\pi.\,( A_1,M_1,\hat t_1,\hat s_1,\psi_1)=\pi\cdot( A_2,M_2,\hat t_2,\hat s_2,\psi_2)$, where $\pi$ a permutation of $O$-names.
  \end{enumerate}
\end{lemma}

Given a play $t$ (from the game-LTS) and a map $\psi$ produced as above, we define its translated $O$-view, $\oview{t}_\psi$, by setting $\oview{\emptytrace}_\psi=\emptytrace$ and:
\begin{align*}
  \oviewpsi{t'\,\Epropret{D}[\vec\be]} &= \begin{cases}
    \emptytrace & \text{if $t'\,\Epropret{D}[\vec\be]$ is complete}\\
    \oviewpsi{t''}\lpropret{D} & \text{otherwise, and $t''\sqsubseteq t'$ ends in last open call of $t'$}
  \end{cases}\\
  \oviewpsi{t'\,\Epropapp{\al}{D}{\vec\be}} &= \lpropapp{\alpha}{D} \quad\qquad \text{if $\psi(\al)\funto\alpha$}\\
  \oviewpsi{t'\,\Eopret{D}[\vec\al]} &=     \oviewpsi{t'}\lopret{D[\vec\alpha]}\quad \text{if $\psi(\vec\al)\funto\vec\alpha$} \\
  \oviewpsi{t'\,\Eopapp{\be}{D}{\vec\al}} &= \begin{cases}
    \emptytrace &\text{if $t'$ is complete}\\
    \oviewpsi{t'}\lopapp{i}{D[\vec\alpha]} & \text{otherwise, with $\psi(\vec\al)=\vec\alpha^{\cdots}$}
    \end{cases}
\end{align*}
where, in the latter case, $\be$ should be the $i$-th $P$-name in $\oviewpsi[\circ]{t'}$. 
We also define the translated top-level view $\topview{t}$ of $t$ by $\topview{\emptytrace}=\emptytrace$ and:
\begin{align*}
  \topview{t_{cp}o\,t_{lo}} &= \topview{t_{cp}\,o} && \text{if }t_{lo}\neq\emptytrace \\
  \topview{t'\,\Epropret{D}[\vec\be]} &= 
    \topview{t''}\lpropret{D} && \text{if $t'\Epropret{D}[\vec\be]$ complete, and $t''\sqsubseteq t'$ ends}\\ &&& \text{in last open call of $t'$}  \\
  \topview{t'\,\Eopapp{\be}{D}{\vec\al}} &= \topview{t'}\lopapp{i}{D[\vec\alpha]} && \text{if $t'$ complete, with $\psi(\vec\al)\funto\vec\alpha$}
\end{align*}
and where $\be$ the $i$-th $P$-name in $\oview{t'}$.
We can show the following.
\begin{lemma}\label{lem:psis}
  Given $t_1,t_2\sqsubseteq t$ that have same-parity length,
  if $\oviewpsi{t_1}=\oviewpsi{t_2}$ then $\oviewpsi[\circ]{t_1}\sim\oviewpsi[\circ]{t_2}$.
\end{lemma}
\begin{proof} Assuming $\oviewpsi{t_1}=\oviewpsi{t_2}$, we show $\oviewpsi[\circ]{t_1}\sim\oviewpsi[\circ]{t_2}$
 by induction on $|\oviewpsi{t_1}|=|\oviewpsi{t_2}|$. The base case (for length $0$) is straightforward. Suppose now $\oviewpsi{t_i}=\hat t\, x$ ($i=1,2$), and do case analysis on $x$:
  \begin{itemize}
    \item If $x=\lpropapp{\alpha}{D}$ then
  $\oviewpsi{t_i}=\lpropapp{\alpha}{D}$ ($i=1,2$). Hence,
  $\oviewpsi[\circ]{t_i}=\Epropapp{\al_i}{D}{\vec\be_i}$ with $\psi(\al_i)=\alpha^{j_i}$, for some $\al_i,\vec\be_i,j_i$, as required.
    \item If $x=\lpropret{D}$ then
      $\oviewpsi[\circ]{t_i}=\oviewpsi[\circ]{t_i'}\Epropret{D}[\vec\be_i]$ ($i=1,2$), for some $\vec\be_i$, and $t_i'\sqsubseteq t_i$ ending in last open call of $t_i$, and $\oviewpsi{t_i}=\oviewpsi{t_i'}\lpropret{D}$. The claim follows from the IH, and the fact that the $\vec\be_i$'s are fresh (and Lemma~\ref{lem:SSL}).
    \item If $x=\lopapp{j}{D[\vec\alpha]}$ then
      $\oviewpsi{t_i}=\oviewpsi{t_i'}\lopapp{i}{D[\vec\alpha]}$, with $t_i=t_i'\Eopapp{\be_i}{D}{\vec\al_i}$.
Moreover, $\oviewpsi[\circ]{t_i}=\oviewpsi[\circ]{t_i'}\Eopapp{\be_i}{D}{\vec\al_i}$. By definition $\be_i$ is the $j$-th $P$-name in $\oviewpsi[\circ]{t_i'}$. Then, the claim follows from the IH and the fact that the $\vec\al_i$'s are fresh  (and Lemma~\ref{lem:SSL}).
    \item If $x=\lopret{D[\vec\alpha]}$ then work as the previous case above. \myqed
     \end{itemize}
\end{proof}
\begin{lemma}\label{lem:oviewpsis}
  For each compatible $\A,\conc,t$:
  \begin{enumerate}
  \item the translation $(\A,\conc,t)^\circ$ is well defined;
  \item if $(\A,\conc,t)^\circ\ni( A,M,\hat t,\hat s,\psi)$ then
    \begin{align*}
      \hat s &= \topview{t}\\      
      M &=\{ \oview{t'}_\psi\mid t'\sqsubseteq t\} \\
      \hat t &=\begin{cases}
        \oview{t}_\psi & \text{if $|t|$ odd or }t=\emptytrace\\
        \oview{t'}_\psi & \text{if $|t|$ even and $t'\sqsubseteq t$ ends in last open call of $t$}
      \end{cases}
    \end{align*}
  \end{enumerate} 
\end{lemma}
\begin{proof} We use mutual induction on $|t|$. For 1,
  we need to check that the requirements imposed on $M$ (when $t$ ends in an $O$-move) are adhered to. Suppose WLOG that $t=t'+\Eopapp{\be}{D}{\vec\al}$ and let $( A',M',\hat t',\hat s',\psi')\in(\A',\conc,t')^\circ$ (where $\A'$ the restriction of $\A$ to names in $t'$).
  If $M$ is not defined on $\hat t'$ then (using also the IH) $(\A,\conc,t)^\circ$ is well defined. On the other hand, if $\nextmove{M}{\hat t'}=x$ then
we need to show that $x=\lopapp{j}{D[\vec\alpha]}$, where $\be$ is the $j$-th $P$-name in $\oviewpsi[\circ]{t'}$, for some $\vec\alpha$. Given $\nextmove{M}{\hat t'}=x$, 
  there is some $t''m\sqsubseteq t'$ such that $\oviewpsi{t''m}=\hat t'x$. 
  We claim that $\oview{t'}\sim\oview{t''}$. By Lemma~\ref{lem:psis}, $\oviewpsi[\circ]{t'}\sim\oviewpsi[\circ]{t''}$.
  Let $t_1=t',t_2=t''$ and suppose
  $\oviewpsi[\circ]{t_i}$ starts with some $\Epropapp{\al_i}{D}{\vec\be_i}$ (for $i=1,2$). Since $\oviewpsi{t_1}=\oviewpsi{t_2}$, there are some $\alpha,j_1,j_2$ such that $\psi(\al_i)=\alpha^{j_i}$.
    Let $t_i'=t_i''o_i\sqsubseteq t_i$  each end in the $O$-move $o_i$ that introduces $\al_i$. If $t_1'=t_2'$ then we are done. Otherwise, the IH and
    the fact that $\psi$ assigns the same name to $\al_1,\al_2$ imply that $\oviewpsi{t_1''}=\oviewpsi{t_2''}$ and, again by Lemma~\ref{lem:psis}, we have that $\oviewpsi[\circ]{t_1''}\sim\oviewpsi[\circ]{t_2''}$. Moreover, using also the IH, we can obtain that $\oviewpsi[\circ]{t_1'}\sim\oviewpsi[\circ]{t_2'}$. Proceeding consecutively this way, we conclude that $\oview{t'}\sim\oview{t''}$. Hence, by $O$-innocence, $\oview{t}\sim\oview{t''m}$, which in turn implies that $\hat s=\lopapp{j}{D[\vec\alpha]}$.

    We now look at 2. The base case is clear.
    Suppose that $t=t'+m$ and let $( A',M',\hat t',\hat s',\psi')\in(\A',\conc',t')^\circ$.
    By IH, $M'=\{\oviewpsi[\psi']{t''}\mid t''\sqsubseteq t'\}$ and therefore $\{\oviewpsi{t''}\mid t''\sqsubseteq t\}=M'\cup\{\oviewpsi{t}\}=M'[\oviewpsi{t}]$.
    We do case analysis on $m$.
    \begin{itemize}
    \item Suppose $m=\Eopret{D}[\vec\al]$.
By IH, $\hat t'=\oviewpsi[\psi']{t'}$.
Moreover, $\oviewpsi{t}= \oviewpsi{t'}\lopret{D[\vec\alpha]}$, with $\psi(\vec\al)=\vec\alpha^{\dots}$, and $M=M'[\hat t'\lopret{D[\vec\alpha]}]=\{\oviewpsi{t''}\mid t''\sqsubseteq t\}$.
Also, $\hat t$ is the trace $\hat t''$ we obtain by translation by looking at the $t''\sqsubseteq t$ ending in the last open call (which must be an $O$-move). By IH, $\hat t''=\oviewpsi{t''}=\oviewpsi{t}$. 
\item   Suppose $m=\Eopapp{\be}{D}{\vec\al}$
  By IH, $\hat t'=\oviewpsi[\psi']{t'}$. If $t'$ is complete, then $\oviewpsi{t}=\emptytrace=\hat t$ and $M=M'$, so we are done. Otherwise, $\oviewpsi{t}=\oviewpsi{t'}\lopapp{i}{D[\vec\alpha]}$, with $\psi(\vec\al)=\vec\alpha^{\dots}$ and $\be$ the $i$-th $P$-name in $\oviewpsi[\circ]{t'}$, and 
  $M=M'[\hat t'\lopapp{i}{D[\vec\alpha]}]=\{\oviewpsi{t''}\mid t''\sqsubseteq t\}$. Also, $\hat t=\hat t'\lopapp{i}{D[\vec\alpha]}=\oviewpsi[\psi]{t'}\lopapp{i}{D[\vec\alpha]}=\oviewpsi{t}$.
    \item Suppose $m=\Epropret{D}[\vec\be]$. If $t$ is complete then the claim is clear. Otherwise, let $t''\sqsubseteq t'$ end in the last open call in $t'$ (which is an $O$-call), and $\hat t''$ the corresponding trace we obtain by translation. By IH, $\hat t''= \oviewpsi{t''}$ and, hence, $\hat t=\hat t'
      \lpropret{D}\overset{\rm IH}{=}\oviewpsi{t''}\lpropret{D}=\hat t''\lpropret{D}=\oviewpsi{t}$.
      Moreover, $M=M'[\hat t'\lpropret{D}]=\oviewpsi{t}=\{\oviewpsi{t''}\mid t''\sqsubseteq t\}$.
    \item Suppose $m=\Epropapp{\al}{D}{\vec\be}$ and let $\psi'(\al)=\alpha^{\dots}$. Then,
      $\hat t=\lpropapp{\alpha}{D}=\oviewpsi{t}$. Moreover, $M=M'[\lpropapp{\alpha}{D}]=\{\oviewpsi{t''}\mid t''\sqsubseteq t\}$. 
    \end{itemize}
        The fact that $\hat s = \topview{t}$ follows from the IH using similar reasoning to the one above.\myqed
\end{proof}

Given a play $t$ and an $O$-name $\al$ of $t$, let us write $t @\al$ for the prefix $t'\sqsubseteq t$ such that the last move of $t'$ is the one introducing $\al$ in $t$.

\begin{lemma}\label{lem:psis2}
  Given a play $t$, even length $t_1,t_2\sqsubseteq t$, names $\al_1,\al_2$ and some $\psi$ from $(t)^\circ$ with $\psi(\al_i)\funto\alpha_i$ ($i=1,2$):
  \begin{enumerate}
    \item $
    \alpha_1= \alpha_2$
iff
    $\oviewpsi{t@\al_1}=\oviewpsi{t@\al_2}\neq\varepsilon\lor
 {t@\al_1}={t@\al_2}$;
\item 
if $\oviewpsi{t_1}=\oviewpsi{t_2}\neq\varepsilon$ 
  then $\oview{t_1}\sim\oview{t_2}$;
  \item if $\alpha_1= \alpha_2$ then
    $\oview{t@\al_1}=\oview{t@\al_2}$.
  \end{enumerate}
  \end{lemma}
\begin{proof}
  For 1, we use the definition of $(\_)^\circ$, and in particular the fact that $\psi$ only assigns old names $\vec\alpha$ when the translated $O$-views are the same, and non-empty, in which case $M$ is defined and forces $O$ to play the same move. On the other hand, the translated $O$-views are empty only if the moves introducing $\al_i$ are top-level, in which case the same $\alpha$ can be assigned two different moves, which means that the last moves in $t@\al_i$ coincide, hence $t@\al_1=t@\al_2$.  
\\
  For 2, we do induction on $|t_1|+|t_2|$.
  If $t_1$ or $t_2$ is empty, the claim is clear.
  Otherwise, by Lemma~\ref{lem:psis} we have that $\oviewpsi[\circ]{t_1}=\oviewpsi[\circ]{t_2}$. 
  Let $p_i$ be the last open $P$-call in $t_i$ and suppose it calls some $\al_i$. By 
  $\oviewpsi{t_1}=\oviewpsi{t_2}$ we have that $\psi$ assigns the same name, say $\alpha$, to $\al_1,\al_2$ and hence, by~1, we have that $\oviewpsi{t@\al_1}=\oviewpsi{t@\al_2}\neq\varepsilon$ or $t@\al_1=t@\al_2$. 
If $\oviewpsi{t@\al_i}\neq\varepsilon$ then, by IH, $\oview{t@\al_1}=\oview{t@\al_2}$ so, in either case, $\oview{t@\al_1}=\oview{t@\al_2}$.
Observe that, for each $i$, $\oview{t_i}=\oview{t@\al_i}\oviewpsi[\circ]{t_i}$.
The claim then follows by repeated application of Lemma~\ref{lem:SSL}.
\\
Claim~3 follows from~1 and~2.
\myqed
\end{proof}

Suppose now we are given plays $t_1,t_2$, and consider $(M_i,\hat t_i,\hat s_i,\psi_i)\in(t_i)^\circ$ ($i=1,2$). By a slight abuse of notation, we shall write $\psi_1\triangleleft\psi_2$ if:
   \[
\forall \alpha,\al_1.\ \psi_1(\al_1)\funto\alpha \implies \exists\al_2.\ \psi_2(\al_2)\funto\alpha \land \oviewpsi[\psi_1]{t_1@\al_1}=\oviewpsi[\psi_2]{t_2@\al_2}.
     \]
     We can show the following results.

\begin{lemma}\label{lem:some}
  Given $t_i'\sqsubseteq t_i$ ($i=1,2$):
  \begin{enumerate}
    \item if $M_1\subseteq M_2$ then $\psi_1\triangleleft\psi_2$;
  \item if $\oview{t_1'}=\oview{t_2'}$ then $\oviewpsi[\psi_1]{t_1'}\sim\oviewpsi[\psi_2]{t_2'}$;
  \item if $\oview{t_1'}=\oview{t_2'}$ and $\psi_1\triangleleft\psi_2$ then $\oviewpsi[\psi_1]{t_1'}=\oviewpsi[\psi_2]{t_2'}$;
  \item if $\oview{t_1'}=\oview{t_2'}$ and $M_1\subseteq M_2$ then $\oviewpsi[\psi_1]{t_1'}=\oviewpsi[\psi_2]{t_2'}$.
    \end{enumerate}
\end{lemma}
\begin{proof}
  For 1, suppose $M_1\subseteq M_2$ and let $\psi_1(\al_1)\funto\alpha$ for some $\al_1,\alpha$. Then, by definition, there is some $t_1''\sqsubseteq t_1$ such that $\oviewpsi[\psi_1]{t_1''}\in M_1$ and the last move in $t_1''$ is an $O$-move introducing $\al_1$, and the corresponding name in $\oviewpsi[\psi_1]{t_1''}$ is $\alpha$. By $M_1\subseteq M_2$ we have $\oviewpsi[\psi_1]{t_1''}\in M_2$, so there is $t_2''\sqsubseteq t_2$ such that
$\oviewpsi[\psi_1]{t_1''}=\oviewpsi[\psi_2]{t_2''}$, so $t_2''$ 
  ends in an $O$-move introducing some name $\al_2$ and such that $\psi_2(\al_2)\funto\alpha$, as required.
  \\
  For 2, we do induction on $|\oview{t_1'}|=|\oview{t_2'}|$.
  The base case is encompassed in that of $t_1'$ being complete, in which case
$t_2'$ is also complete and
  the claim trivially holds. Similarly if $t_1'=t_1''m$ with $t_1''$ complete.
  Now let $t_i'=t_i''m$, with $t_i',t_i''$ not complete, and do case analysis on $m$.
  If $m$ is a $P$-return then the claim follows from the IH; if $m$ is a $P$-application, then it directly follows from the definition of $\oviewpsi[\psi_i]{t_i'}$.
  \\
  If $m$ is an $O$-move then $\oview{t_i'}=\oview{t_i''}m$ and by IH $\oviewpsi[\psi_1]{t_1''}\sim\oviewpsi[\psi_2]{t_2''}$. If $m=\Eopapp{\be}{D}{\vec\al}$, so 
  $\oviewpsi[\psi_i]{t_i'}$ ends in some $\lopapp{j_i}{D[\vec\alpha_i]}$,
  by $\oview{t_1'}=\oview{t_2'}$ we obtain $j_1=j_2$. As the $\vec\alpha_i$'s are fresh for $\oviewpsi[\psi_i]{t_i''}$, from Lemma~\ref{lem:SSL} we obtain $\oviewpsi[\psi_1]{t_1'}\sim\oviewpsi[\psi_2]{t_2'}$.
  Similarly if $m=\Eopret{D}[\vec\al]$.
%

  For 3, we again do induction on $|\oview{t_1'}|=|\oview{t_2'}|$. As above, let us assume that $t_i'=t_i''m$, with $t_i',t_i''$ not complete, and do a case analysis on $m$. If $m$ is a $P$-return then we simply use the IH.
  \\
%
  If $m=\Epropapp{\al}{D}{\vec\be}$ then $\oviewpsi[\psi_i]{t_i'}=\lpropapp{\alpha_i}{D}$ with $\psi_i(\al)\funto\alpha_i$. Note that
$\oview{t_1'}=\oview{t_2'}$ implies that
$\oview{t_1'@\al}=\oview{t_2'@\al}$. Hence, by IH, $\oviewpsi[\psi_1]{t_1'@\al}=\oviewpsi[\psi_2]{t_2'@\al}$.
As $\psi_1\triangleleft\psi_2$, there is $\al'$ such that $\psi_2(\al')\funto\alpha_1$ and $\oviewpsi[\psi_1]{t_1@\al}=\oviewpsi[\psi_2]{t_2@\al'}$, and therefore 
$\oviewpsi[\psi_2]{t_2@\al}=\oviewpsi[\psi_2]{t_2@\al'}$. By Lemma~\ref{lem:psis2} we obtain that $\alpha_1=\alpha_2$, as required.
\\
If $m=\Eopret{D}[\vec\al]$ then
$\oviewpsi[\psi_i]{t_i'}=\oviewpsi[\psi_i]{t_i''}\lopret{D[\vec\alpha_i]}$, with $\psi_i(\vec\al)\funto\vec\alpha_i$. By IH, $\oviewpsi[\psi_1]{t_1''}=\oviewpsi[\psi_2]{t_2''}$, so it suffices to show that $\vec\alpha_1=\vec\alpha_2$.
Let us pick some $\al_j\in\vec\al$, so $\psi_i(\al_j)\funto\alpha_{i,j}$. By hypothesis, there is $\al_j'$ such that $\psi_2(\al_j')\funto\alpha_{1,j}$ and $\oviewpsi[\psi_1]{t_1''}=\oviewpsi[\psi_2]{t_2@\al_j'}$. As $\oviewpsi[\psi_1]{t_1''}=\oviewpsi[\psi_2]{t_2''}$, we have $\oviewpsi[\psi_2]{t_2@\al_j'}=\oviewpsi[\psi_2]{t_2''}$ so $\alpha_{1,j}=\alpha_{2,j}$, by Lemma~\ref{lem:psis2}.
Similarly if $m=\Eopapp{\be}{D}{\vec\al}$.
\\
Claim~4 follows from~1 and~3.
\myqed
\end{proof}

\begin{lemma}\label{lem:OVsMs}
  If $\OVs{(t_1)}\subseteq\OVs{(t_2)}$ then for each $(M_1,\hat t_1,\hat s_1,\psi_1)\in(t_1)^\circ$ there is 
  $(M_2,\hat t_2,\hat s_2,\psi_2)\in(t_2)^\circ$ such that $\psi_1\triangleleft\psi_2$ and $M_1\subseteq M_2$. 
\end{lemma}
\begin{proof}
  Note that $M_1\subseteq M_2$ follows from $\psi_1\triangleleft\psi_2$ and Lemmas~\ref{lem:oviewpsis} and~\ref{lem:some}, and the fact that $\OVs(t_1)\subseteq\OVs(t_2)$. We use induction on $|t_1|$ to show $\psi_1\triangleleft\psi_2$. If $t_1$ is empty, then the claim is trivial.
Otherwise, let $t_1=t_1'm$ and let $(M_1',\hat t_1',\hat s_1',\psi_1')\in(t_1')^\circ$.
\\
Suppose $m=\Eopret{D}[\vec\al]$; the case where $m=\Eopapp{\be}{D}{\vec\al}$ is treated similarly. If $M_1'$ is defined on $\hat t_1'$ then $\psi_1=\psi_1'$ and the claim holds. Otherwise, $\psi_1=\psi_1'[\vec\al\mapsto\vec\alpha^i]$ for fresh $\vec\alpha$ and some $i$. By hypothesis, there is a permutation $\pi$ and some $t_2'm\sqsubseteq \pi\cdot t_2$ such that $\oview{t_1'}=\oview{t_2'}$. By IH, there is $(M_2',\hat t_2',\hat s_2',\psi_2')\in(\pi\cdot t_2)^\circ$
with $\psi_1'\triangleleft\psi_2'$. Hence, by Lemma~\ref{lem:some}, $\oviewpsi[\psi_1']{t_1'}=\oviewpsi[\psi_2']{t_2'}$.
Suppose $\psi_2'(\vec\al)\funto\vec\alpha'$,
let $\pi'=(\vec\alpha\ \vec\alpha')$ and set $\psi_2=\pi'\cdot\pi^{-1}\cdot\psi_2'$.
Note that $\vec\alpha$ are fresh for $\psi_1'$. Suppose there is some $\alpha_j'\in\vec\alpha'$ that is not fresh for $\psi_1'$, e.g.\ $\psi_1'(\al_1')=\alpha_j'$ for some $\al_1'$. We then must have some $\al_2'$ such that $\psi_2'(\al_2')\funto\alpha_j'$ and
$\oviewpsi[\psi_1']{t_1'@\al_1'}=\oviewpsi[\psi_2']{(\pi\cdot t_2)@\al_2'}$. 
By Lemma~\ref{lem:psis2} we then obtain $\oviewpsi[\psi_2']{(\pi\cdot t_2)@\al_2'}=\oviewpsi[\psi_2']{(\pi\cdot t_2)@\al_j}=\oviewpsi[\psi_2']{t_2'}$. Recall that $\oviewpsi[\psi_1']{t_1'}=\oviewpsi[\psi_2']{t_2'}$, hence $\oviewpsi[\psi_1']{t_1'}=\oviewpsi[\psi_1']{t_1'@\al_1'}$.
But then, by Lemma~\ref{lem:psis2} again, we have that $\alpha_j'=\alpha_j$, a contradiction to $\alpha_j$ being fresh. Hence, $\vec\alpha'$ are fresh for $\psi_1'$, and thus $\psi_1'=\pi'\cdot\psi_1'$, which implies $\psi_1'\triangleleft\pi'\cdot\psi_2$ and therefore
$\psi_1\triangleleft\pi'\cdot\psi_2$. Since $\pi$ only permutes names from the extended LTS, we can also deduce that $\psi_1\triangleleft\pi^{-1}\cdot\pi'\cdot\psi_2=\psi_2$. 
Finally, from $(M_2',\hat t_2',\hat s_2',\psi_2')\in(\pi\cdot t_2)^\circ$
we obtain that $(M_2',\hat t_2',\hat s_2',\pi^{-1}\cdot\psi_2')\in(t_2)^\circ$ and thus
$(\pi'\cdot M_2',\pi'\cdot \hat t_2',\pi'\cdot \hat s_2',\psi_2)\in(t_2)^\circ$, as required.
\\
If $m$ is a $P$-move, then the claim follows directly from the IH (on $t_1'$). \myqed
\end{proof}

\begin{lemma}\label{lem:MsOVs}
Given $(M_i,\hat t_i,\hat s_i,\psi_i)\in(t_i)^\circ$ (for $i=1,2$), if $M_1\subseteq M_2$ and $\hat s_1=\hat s_2$ then $\OVs{(t_1)}\subseteq\OVs{(t_2)}$.
\end{lemma}
\begin{proof}
  We show that, for any $t_1'\sqsubseteq t_1$, we have $\OVs(t_1')\subseteq\OVs(t_2)$, by induction on $|t_1'|$.
  Base case is clear.
\\
  Now let $t_1'=t_1''m_1$ with $m_1$ an $O$-move so, by IH, we have $\OVs(t_1'')\subseteq\OVs(t_2)$. We do case analysis on $m_1$.
Suppose $m_1=\Eopapp{\be_1}{D}{\vec\al_1}$ and $\psi_1(\vec\al_1)\funto\vec\alpha$, and assume $t_1''$ is not complete; the case of 
$m_1=\Eopret{D}[\vec\al_1]$ is dealt with similarly. Then,
by definition, 
$\nextmove{M_1}{\oviewpsi[\psi_1]{t_1''}}=\lopapp{j}{D[\vec\alpha]}$ and, since $M_1\subseteq M_2$,
$\nextmove{M_2}{\oviewpsi[\psi_1]{t_1''}}=\lopapp{j}{D[\vec\alpha]}$. By IH we have $\oview{t_1''}\in\OVs(t_2)$, so let $t_2''\sqsubseteq\pi\cdot t_2$ (for some $\pi$) such that $\oview{t_1''}=\oview{t_2''}$.
By Lemma~\ref{lem:some}, $\oviewpsi[\psi_1]{t_1''}=\oviewpsi[\pi\cdot\psi_2]{t_2''}$. By hypothesis, 
$\nextmove{M_2}{\oviewpsi[\pi\cdot\psi_2]{t_2''}}=\lopapp{j}{D[\vec\alpha]}$, so there is 
$m_2$ such that $t_2''m_2\sqsubseteq\pi\cdot t_2$ and $\oviewpsi[\pi\cdot\psi_2]{t_2''m_2}=\oviewpsi[\psi_1]{t_1''}\lopapp{j}{D[\vec\alpha]}$, hence $m_2=\Eopapp{\be_2}{D}{\vec\al_2}$.
By definition, $\be_i$ is the $j$-th $P$-name in $\oview{t_i''}$, thus $\be_1=\be_2$. Using also Lemma~\ref{lem:SSL}, we have $\oview{t_1''}m_1\sim\oview{t_2''}m_2$, as required.
Finally, suppose $t_1''$ 
is complete and pick $t_2''\sqsubseteq\pi\cdot t_2$ (for some $\pi$) such that $\oview{t_1''}=\oview{t_2''}$.
By $\hat s_1=\hat s_2$ and Lemma~\ref{lem:oviewpsis}, and since  
$(M_2,\hat t_2,\hat s_2,\pi\cdot\psi_2)\in(\pi\cdot t_2)^\circ$, there is $m_2$ such that $t_2''m_2\sqsubseteq \pi\cdot t_2$ and $m_2=\Eopapp{\be_2}{D}{\vec\al_2}$, and we can conclude as above.
\\
Suppose $t_1$ ends in $m_1=\Epropapp{\al}{D}{\vec\be}$ and let $t_1''=t_1'@\al\sqsubseteq t_1'$. By IH, there is $t_2''\sqsubseteq\pi\cdot t_2$ (for some $\pi$) such that $\oview{t_1''}=\oview{t_2''}$, so in particular $t_2''$ ends in $o$. By Lemma~\ref{lem:some}, $\oviewpsi[\psi_1]{t_1''}=\oviewpsi[\pi\cdot\psi_2]{t_2''}$, hence there is some $\alpha$ such that $\psi_1(\al_1)\funto\alpha$ and
$(\pi\cdot\psi_2)(\al_2)\funto\alpha$.
Since $M_1\subseteq M_2$, we have $\lpropapp{\alpha}{D}\in M_2$ and hence there must be $s_2'm_2\sqsubseteq\pi\cdot t_2$ with $m_2=\Epropapp{\al'}{D}{\vec\be'}$ such that $(\pi\cdot\psi_2)(\al')\funto\alpha$.
Hence, by Lemma~\ref{lem:psis2}, $\oview{t_1''}=\oview{t_2''}=\oview{s_2'@\al'}$.
Using Lemma~\ref{lem:SSL}, we obtain $\oview{t_1''}m_1\sim\oview{s_2'@\al'}m_2$, as required.
If $t_1$ ends in some $\Epropret{D}{\vec\be}$ then we work similarly, using also the fact that $\hat s_1=\hat s_2$ in case $m_1$ is top-level. \myqed
\end{proof}

\begin{definition}
   We define a translation of configurations from the game-LTS to the plain LTS as follows:
  \begin{align*}
    (\oconf{\A}{\conc}{K}{t}{V}{\vec\be})^\circ &=
                                                 \{\oconf{ A}{M}{\hat K}{\hat t}{\hat V}{\vec v} \mid \exists \psi.\ ( A,\hat M,\hat t,\hat s,\psi)\in(\A,\conc,t)^\circ\}\\
    (\pconf{\A}{\conc}{K}{t}{e}{V}[\vec\al])^\circ &=
\{\pconf{ A}{M}{\hat K}{\hat t}{\hat e}{\hat V}
\mid \exists \psi.\ ( A,\hat M,\hat t,\hat s,\psi)\in(\A,\conc,t)^\circ\} 
    \end{align*}
where       $(\hat e,\hat K,\hat V,\vec v)=\psi(e, K,\conc(V),\conc(\vec\be))$.
\end{definition}

We can now show the following result. Below we write $x$ for $\tau$ or a move $m$, and $\chi$ for $\tau$ or a move $\eta$.
  \begin{lemma}\label{lem:twoLTS}
    Let $\C$ be a reachable configuration of the game-LTS and let $C\in(\C)^\circ$. Then:
    \begin{itemize}
      \item if $\C\xrightarrow{x}\C'$ then $\exists C'\in(\C')^\circ,\chi.\ C\xrightarrow{\chi}C'$;
      \item if $C\xrightarrow{\chi}C'$ then $\exists \C',x.\ \C\xrightarrow{x}\C'\land C'\in(\C')^\circ$;
      \end{itemize}
      and either $\chi=x=\tau$, or $\chi,x$ are of the same O/P and call/return shape, with their components related as in $C,\C$ and $C',\C'$. 
  \end{lemma} 

Recall that we  write $\C_e$ for the initial configuration for term $\vdash e:T$ in the game-LTS, and 
$C_e$ for the one in the plain LTS.

  \begin{proposition}
    Given $\vdash e_1,e_2:T$, $\OVTLs(\C_{e_1})=\OVTLs(\C_{e_2})$ iff $\sem{e_1}=\sem{e_2}$.
  \end{proposition}
  \begin{proof}
    Suppose that $\OVTLs(\C_{e_1})=\OVTLs(\C_{e_2})$, and let $(t_1,M_1)\in\sem{e_1}$. We show that $(t_1,M_1)\in\sem{e_2}$. Let $C_{e_1}\xtoo{t_1}C'_{e_1}$, with $C'_{e_1}=\oconf{A}{M_1}{\emptyseq}{\hat t_1}{V}{\vec v}$. By Lemma~\ref{lem:twoLTS}, $\C_{e_1}\xtoo{s_1}\C_{e_1}'$ with $(\C_{e_1}')^\circ\ni C'_{e_1}$. By hypothesis, $\C_{e_2}\xtoo{s_2}\C_{e_2}'$ with $\OVs(s_1)=\OVs(s_2)$, and $s_1,s_2$ both top-linear.
WLOG, assume that $\oview{s_1}=\oview{s_2}$.
By Lemma~\ref{lem:twoLTS} again, $C_{e_2}\xtoo{t_2}C'_{e_2}$ with $C'_{e_2}=(\C'_{e_2})^\circ$.
By Lemma~\ref{lem:OVsMs},
the last claim of Lemma~\ref{lem:twoLTS} and the fact that $\oview{s_1}=\oview{s_2}$, we can pick $C'_{e_2}$ in such a way that its $M$-component is $M_1$ and $t_1=t_2$.
\\
Conversely, suppose $\sem{e_1}=\sem{e_2}$ and let $t_1$ be a complete (top-linear) play of $\C_{e_1}$. We need to show $\OVs{(t_1)}\in\OVTLs(\C_{e_2})$.
Considering the transition sequence $\C_{e_1}\xtoo{t_1}\C_1$, by Lemma~\ref{lem:twoLTS} there is
$C_{e_1}\xtoo{\hat t}C_1$ with $C_1\in(\C_1)^\circ$ and final $M$-component $M$. By hypothesis then, $C_{e_2}\xtoo{\hat t}C_2$ with final $M$-component $M$. Using Lemma~\ref{lem:twoLTS} again we obtain $\C_{e_2}\xtoo{t_2} \C_2$ with $C_2\in(\C_2)^\circ$ and such that $\hat t$ is the corresponding translation of both $t_1$ and $t_2$ (so $t_2$ also top-linear). Then, by Lemma~\ref{lem:MsOVs}, we have $\OVs(t_1)=\OVs(t_2)$. \myqed
\end{proof}



  \end{document}